\documentclass{article}

\author{}

\usepackage{amsmath}
\usepackage{amsthm}
\usepackage{amsfonts}
\usepackage{amssymb}
\usepackage{latexsym}
\usepackage{setspace}
\usepackage[dvips]{graphicx}
\usepackage{epsfig}
\usepackage{graphics}
\usepackage{mathrsfs}
\usepackage{color}

\newtheorem{theorem}{Theorem}[section]

\newtheorem{proposition}{Proposition}[section]

\newtheorem{remark}{Remark}[section]

\begin{document}

\large

\title{Systemic Risk with Exchangeable Contagion: \\ Application to the European Banking System
\thanks{ The authors would like to thank Robin Treber for excellent research assistance. We also thank participants in the Conference on High-Dimensional Dependence and Copula in Bejing, in the Workshop on Dependence Models and Risk, in Bozen, and in the Conference on Systemic Risk and Contagion at the University of Bologna for useful comments.}}
\author{Umberto Cherubini\footnote{Corresponding author: umberto.cherubini@unibo.it.}, Sabrina Mulinacci
\\ University of Bologna - Department of Statistics}

\date{}
\maketitle

\begin{abstract} We propose a model and an estimation technique to distinguish systemic
risk and contagion in credit risk. The main idea is to assume, for
a set of $d$ obligors, a set of $d$ idiosyncratic shocks and a
shock that triggers the default of all them. All shocks
are assumed to be linked by a dependence relationship, that in
this paper is assumed to be exchangeable and Archimedean. This
approach is able to encompass both systemic risk and contagion,
with the Marshall-Olkin pure systemic risk model and the
Archimedean contagion model as extreme cases. Moreover, we show
that assuming an affine structure for the intensities of
idiosyncratic and systemic shocks and a Gumbel copula, the
approach delivers a complete multivariate distribution with
exponential marginal distributions. The model can be estimated by
applying a moment matching procedure to the bivariate marginals.
We also provide an easy visual check of the good specification of
the model. The model is applied to a selected sample of banks for
8 European countries, assuming a common shock for every country.
The model is found to be well specified for 4 of the 8 countries.
We also provide the theoretical extension of the model to the
non-exchangeable case and we suggest possible
avenues of research for the estimation.\\


\end{abstract}

\section{Introduction}
The purpose of this paper is to draw a line between systemic risk
and contagion, and to design a method to measure the relative
contribution of contagion and systemic risk to the dependence structure of a set of credit
exposures. The statistical problem of disentangling systemic risk and contagion is of utmost relevance for economic policy.
In several problems, such as pollution regulation or banking, recognizing systemic risk, as an event independent of the agents, or contagion,
that is a system wide event triggered by one of them, is a major discriminant factor to decide whether the effects should be charged to the
community at large or to the individual agents.\\
The problem is involved for two reasons. The first, already
addressed in many studies on the subject, is that the systemic
risk factor is not observed and for credit risk applications we
are only able to extract the marginal survival probabilities from
the market. The second reason, that is the subject of this
paper, is whether such dependence is explained by the presence of a systemic risk factor only, or of some infectious elements in the system.\\
To make the problem clear, assume that we are allowed to observe
the systemic risk factor, and we are able to appraise the
probability of a systemic crisis. In this situation, the question
would naturally arise whether the systemic shock is independent of
the other events triggering the default of each component of the
set. Answering this question on practical grounds would be
obviously easy in this setting, and the dependence between
idiosyncratic and systemic triggers of default could be estimated
in the
usual way, e.g. using copulas.\\
Notice, however, that even if dependence with the systemic risk
factor were observed, a problem of interpretation of these results
on theoretical grounds would arise concerning how it would affect
the observed dependence structure of the components in the system.
In fact, the presence of a systemic factor is sufficient to induce
dependence among the components of the system and between each
component and the systemic event. In other words, this dependence
shows up even if the systemic shock is independent of the
idiosyncratic ones. Intuitively, if the idiosyncratic default
drivers were linked to the systemic risk trigger by a dependence
relationship, the degree of dependence in the system would be even
stronger. \\
In this paper we propose a model to represent these two sources of
dependence in a tractable way. The idea of our model is very
simple. Given a cluster of $d$ obligors, we assume that the system
is subject to a set of $d+1$ shocks, one of which is common to all
the components and leads to simultaneous default of all the
obligors in the cluster, while the others are responsible for the
default of each component. Assume further that the times of these
shocks are linked by a copula function of dimension $d+1$. It is
immediate to see that this model encompasses the two extreme cases
of pure systemic risk and pure contagion. Namely, in the case of a
product copula for the $d+1$ occurrence times of the shocks, we
obtain a model with bivariate Marshall-Olkin marginal
distributions, representing the pure systemic risk model. The
opposite case arises when the systemic shock has zero probability,
so that we have a standard survival copula model corresponding to
pure contagion. Allowing for positive probability of a systemic
crisis and for dependence between this risk and the idiosyncratic
credit drivers would then allow to design models that represent
both systemic risk and contagion. In these models, the task of
disentangling the two
is a relevant question.\\
In its simplest version, our paper assumes the standard
restriction of the choice of exchangeable copulas for the credit
risk drivers. This means that each idiosyncratic factor is assumed
to be linked by the same dependence structure to the systemic risk
factor and to the other idiosyncratic ones. We also show that
further restrictions change the copula model in a multivariate
distribution model with exponential margins. If the model may seem
restrictive, on practical grounds we provide a methodology to
verify if the assumption is borne out by the data. Moreover, on
theoretical grounds, we will also provide the theoretical
development
for the non-exchangeable version of model. \\
The plan of the paper is as follows. After reviewing the relevant
literature, in Section \ref{model} we motivate and describe in full
generality our credit risk model for a basket of issuers, we
discuss its main properties and the restrictions that change the
copula model in a multivariate distribution with exponential
marginals. In Section \ref{extension} we discuss the theoretical features of the
extension to non-exchangeable dependence of the credit risk
drivers. Finally, in Section \ref{sector} we illustrate our application to
the banking system of a set of countries of the Euro area. In
Section \ref{conclusion} we report conclusions and a discussion of the main
issues left for future research.
\subsection{Related literature}
Our paper is related to a large literature on the measurement of systemic risk and contagion, even though to the best of our knowledge it is the first attempt to disentangle the two. Leaving aside any hope of being exhaustive, we may provide a taxonomy of the main contributions according to the structure of models and the data used.\\
As for the methodology involved, a first class of models are based
on the application of Granger causality, and related concepts, to
the prices of financial assets (Billio et al.,2012). A second set
of models is based on the network representation of the
relationships among financial institutions (Diebold and Yilmaz,
2011). A third approach is based on the theory of risk measures
applied to systemic risk and contagion. Models in this class are
based on the measurement of expected losses conditional on an
extreme scenario of some systemic risk factor. The technique is
the same as expected shortfall, with the difference of
conditioning with respect to a systemic variable. These measures
are called Marginal Expected Shortfall, MES (Acharya et al.,
2010), and CoVaR (Adrian and Brunnermeier, 2011). Cherubini and
Mulinacci (2014) give conditions to ensure that \emph{coherence}
requirements be met, and propose examples of measures in this
class based on copula functions.\\ Coming to the kind of data that
are used in the empirical analysis, we may distinguish between
applications that rely on the analysis of market prices, and those
that use flows and balance sheet data. The first choice use equity
stock prices (Billio at al., 2012) or volatilities (Diebold and
Yilmaz, 2011), credit spreads of bonds and credit derivatives
(Baglioni and Cherubini, 2013). With this choice, the focus is on
measurement of the effects of systemic risk and contagion, in
terms of future cash flows and the default probability that are
implied in market quotes. The second choice exploits flows among
the financial intermediaries and the focus is more on the means
that explain propagation of the shocks through
the financial intermediation system. Here the analysis is focussed on flows in the interbank market (Bonaldi, Hortacsy and Kastl, 2013) or on several layers representing other markets (Bargigli et al., 2013), or else on balance sheet indexes such as leverage (Brownlees and Engle, 2010). All these proxies are used as measures of the strength of contagion in the system.\\
Our paper uses the default probability extracted from CDS and their dependence structure in order to recognize how much of this dependence is
due to relationships among the components of the system, as in network based models, and how much of the co-movement is due to the presence of a systemic risk factor,
as in systemic risk models. Moreover, to the best of our knowledge, this is the first attempt to include a dependence structure between each component and the systemic shock, although being the same dependence structure in both cases. \\

\section{The model}\label{model}
Here we introduce the motivation of the model, and its basic
setting. The idea is that in a system of $d$ components, the
lifetime of each of them can come to an end either for
idiosyncratic or systemic shocks, as in a standard Marshall-Olkin
setting. Differently from that model, in which all shocks are
assumed to be independent, here the idiosyncratic components are
infectious. Idiosyncratic defaults can be associated, and they may
also represent triggers of the systemic shock, leading to default
of the whole
system.\\
Technically, let $(\Omega,\mathcal F,\mathbb P)$ be a probability
space with a $d+1$-vector, $(X_0,X_1,\ldots ,X_d)$ whose
components have $[0,+\infty)$ as support. $X_0$ denotes the
arrival time of the systemic shock and $(X_1,\ldots ,X_d)$ are
those of the idiosyncratic ones. We assume that the joint survival
dependence structure is represented by a strict Archimedean
copula, that is
$$\bar F(x_0,x_1,\ldots ,x_d)=\psi\left (\psi^{-1}(\bar F_0(x_0))+\cdots +\psi^{-1}(\bar F_d)\right )$$
for $(x_0,\ldots ,x_d)\in[0,+\infty )^{d+1}$, where $\bar F_i$
(that is assumed to be continuous and strictly decreasing) is the
marginal survival function of $X_i$ and $\psi$ is the generator of
a strict $d+1$-dimensional Archimedean copula. We recall that
$\psi$ is the generator of a $d+1$-Archimedean copula if and only
if $\psi:[0,+\infty )\rightarrow [0,1]$ is $d+1$-monotone on
$[0,+\infty )$ that is
\begin{itemize}
 \item it is differentiable on $(0,+\infty)$ up to order $d-1$ and the derivatives satisfy
 $(-1)^k\psi^{(k)}(x)\geq 0$ for $k=0,1,\ldots ,d-1$ and $x\in (0,+\infty)$
 \item $(-1)^{d-1}\psi^{(d-1)}$ is non-increasing and convex in $(0,+\infty)$.
\end{itemize}
(see McNeil and Ne\v{s}lehov\'{a}, 2009, for more details on multidimensional Archimedean copulas).\medskip

Since we restrict ourselves to the strict case, we assume
$\psi(x)>0$ for all $x\in [0,+\infty)$.
Let us define
$$\tau_k=\min \{X_0,X_k\},\, k=1,\ldots ,d.$$
This is the standard Marshall-Olkin setting in which the only
common shock taken into account is the one affecting all the
components in the set. Of course, other specifications are
possible, including models with more than one common shock,
affecting selected subsets of the components (see, for all,
Durante, Hofert and Scherer, 2010). The observed default times
$\tau _k$ represent the first arrival time between a common
(systemic) shock affecting all the system and the idiosyncratic
shocks.   We then add an Archimedean type of dependence among the
arrival times of the shocks, in order to represent contagion.
\medskip

The joint survival function of the random vector ${\bf
\tau}=(\tau_1,\ldots ,\tau _d)$ can be easily recovered
\begin{equation}\label{joint}
 \bar F_{{\bf \tau}}(t_1,\ldots ,t_d)=\psi\left (\psi^{-1}(\bar F_0(\max_{1\leq k\leq d}\{t_k\}))+\sum_{k=1}^d\psi^{-1}(\bar F_k(t_k))\right )
\end{equation}
for $t_1,\ldots ,t_d\in [0,+\infty )^d$, while the marginal
survival functions are
\begin{equation}\label{marginal}
 \bar F_{\tau_k}(t)=\psi\left (\psi^{-1}(\bar F_0(t))+\psi^{-1}(\bar F_k(t))\right )=\psi\left (H_{0,k}(t)\right ),\, t\in [0,+\infty)
 \end{equation}
 where $H_{0,k}(x)=\psi^{-1}(\bar F_0(x))+\psi^{-1}(\bar
 F_k(x))$.\\
It is also easy to extract the copula function of the observed
default times
\begin{proposition}\label{Mulinacci}
The survival copula $\hat C$ of the vector of default times ${\bf
\tau}$ is, for ${\emph{\bf u}}\in [0,1]^d$,
\begin{equation}\label{copula}
\hat C({\emph{\bf u}})=\sum_{j=1}^d\psi\left
(\psi^{-1}(u_j)+\sum_{k=1,k\neq j}^d D_k\circ \psi^{-1}(u_k)\right
)
 {\bf 1}_{A_j}({\emph{\bf u}})
\end{equation}
where $D_k(x)=\psi^{-1}\circ \bar F_k \circ H_{0,k}^{-1}(x)$ and
$$A_j=\left\{{\emph{\bf u}}\in[0,1]^d:\max_{1\leq i\leq d}\{ H_{0,i}^{-1}\circ \psi^{-1}(u_i)\}=
H_{0,j}^{-1}\circ \psi^{-1}(u_j)\right \}$$ with the convention
that if ${\emph{\bf u}}$ satisfies the required condition for more
than one index $j$, it is assumed to belong to the $A_j$ with the
smallest index $j$.
\end{proposition}
\begin{proof} See Appendix \ref{appendix}.
\end{proof}

\subsection{A multivariate distribution with contagion and exponential
marginals}\label{exponential} In practical applications it is
common to represent and calibrate default times by exponential
distributions
\begin{equation}\label{Muli3}\bar F_{\tau_k}(x)=\exp(-\mu_kx),\end{equation} where $\mu_k$ denotes the
intensity parameter. We now discuss which restrictions can be
imposed on the model in order to transform the copula model above
in a multivariate distribution with exponential marginals.

Starting from the copula model illustrated, constructing such
multivariate distribution would imply a data generating process
such that: i) the distortion functions $D_k(x)$ are linear; ii)
the dependence is represented by a Gumbel copula.

\subsubsection{Linear distortion} A possible assumption about
functions $\psi^{-1}(\bar F_i(x))$, in the spirit of the paper by
Muliere and Scarsini (1987), is that they are all proportional to
the same function $K(x)$: that is, $\psi^{-1}(\bar F_i(x))=\lambda
_iK(x)$ for $\lambda _i>0$, for $i=0,1,\ldots ,d$. This is
equivalent to $D_i(x)=(1-\alpha _i)x$ where 
$$\alpha
_i=\frac{\lambda _0}{\lambda _i+\lambda _0}\in [0,1)$$ and the
obtained copula is independent of $K$.\\In the more specific case
in which $\psi$ is completely monotone (that is $\psi$ is the
Laplace transform of some positive random variable), we recover
the Scale-Mixture of Marshall-Olkin distributions and copula
models (SMMO) studied in Li (2009). The exchangeable case of SMMO
model is studied in Mai and Scherer (2013) where it is applied to
the pricing of CDOs.

\subsubsection{The Gumbel case} A further restriction to yield marginal exponential distributions is to consider the case in
which $\psi$ is the Gumbel generator, that is
$\psi(x)=e^{-x^{\frac 1\theta}}$, $\theta\geq 1$. Now, equations
(\ref{joint}), (\ref{marginal}) and (\ref{copula}) take the form
$$\bar F_{{\bf \tau}}(t_1,\ldots ,t_d)=\exp\left\{-\left (\lambda _0K\left (\max_{1\leq i\leq d}\{t_i\}\right )+
\sum_{k=1}^d\lambda _kK(t_k)\right )^{\frac 1\theta}\right\}$$

\begin{equation}\label{Mulinacci2}\bar F_{\tau _k}(t)=\exp\left (-(\lambda_0+\lambda_k)^{\frac
1\theta} K^{\frac 1\theta}(t)\right )\end{equation}

$$\hat C({\emph{\bf u}})=\sum_{j=1}^d\exp\left\{-\left [(-\ln u_j)^\theta+\sum_{k=1,k\neq j}^d(1-\alpha _k)(-\ln u_k)^\theta\right ]^{\frac 1\theta}\right\}
{\bf 1}_{A_j}({\emph{\bf u}})$$ Notice that, in this case, $\psi$
is the Laplace transform of an $\frac 1\theta$-stable distributed
random variable and so it represents a specification of the SMMO
model of Li (2009).\bigskip

Notice that setting $K(t)=t^\theta$ in (\ref{Mulinacci2}) yields
exponential marginals as required
\begin{equation}\label{Marg_Int}\mu_k =
(\lambda_0+\lambda_k)^{\frac 1\theta} \end{equation} where $\mu_k$
is the intensity in equation (\ref{Muli3}).

\subsection{Properties of the model}

The main feature of our model, right from the most general
setting, is to increase the degree of dependence among the default
times, both with respect to the standard Archimedean copula
without any systemic risk factor and the Marshall-Olkin copula in
which the systemic risk factor is independent of the others.

The dependence structure of the model encompasses both the
sensitivity of the default times to the systemic shock, and the
dependence among the shocks, represented by Archimedean copulas.
Both these elements interact to determine the dependence among
default times.

In the general setting, the Kendall's tau $\tau _{i,k}$ measuring the dependence of the pair of default times $(\tau _i,\tau _k)$
can be written as
$$\tau_{j,k}= \tau^{\psi}+ 4\int_0^{\infty}(\psi^\prime (x))^2\cdot T(x)dx $$
where $\tau^{\psi}$ denotes the Archimedean Kendall's tau corresponding to the generator $\psi$ and
$$T(x) =\psi^{-1}\circ \bar F_0 \circ \left (\psi^{-1}\circ \bar F_0+\psi^{-1}\circ \bar F_j+\psi^{-1}\circ \bar F_k \right)^{-1}(x) $$
where we refer the reader to Mulinacci (2014) for the derivation.
Notice that if we are interested in representing the dependence
structure between the systemic shock and default times, we have that the Kendall's tau $\tau_{j,0}$
of the pair $(\tau _j, X_0)$ is
$$\tau_{j,0}= \tau^{\psi}+ 4\int_0^{\infty}(\psi^\prime (x))^2\cdot \left (\psi^{-1}\circ \bar F_0 \circ \left (\psi^{-1}\circ \bar F_0+\psi^{-1}\circ \bar F_j \right)^{-1}(x) \right )dx $$
The first term is simply the Kendall's tau of the Archimedean copula
used in the analysis, while the other term, that is more complex,
involves both the generator of the Archimedean copula and the
relative relevance of systemic and idiosyncratic shocks.\\
In the multivariate distribution arising with linear distortions
and the Gumbel copula in the model, these relationships simplify
substantially. In fact, let $\hat C_{j,k}(u,v)$ be the general
marginal $2$-copula,
$$\begin{aligned}\hat C_{j,k}(u,v)&=\exp\left\{-\left [(-\ln u)^\theta+(1-\alpha _k)(-\ln v)^\theta\right ]^{\frac 1\theta}\right\}
{\bf 1}_{\left\{\alpha_j\psi^{-1}(u)\geq
\alpha_k\psi^{-1}(v)\right\}}+\\
&+\exp\left\{-\left [(1-\alpha _j)(-\ln u)^\theta+(-\ln
v)^\theta\right ]^{\frac 1\theta}\right\} {\bf
1}_{\left\{\alpha_j\psi^{-1}(u)< \alpha _k\psi^{-1}(v)\right\}}
\end{aligned}$$
Since this family of copulas represents a particular specification
of the Archimax copulas of Cap\'{e}ra\`{a} et al. (2000) and of
the Archimedean-based Marshall-Olkin copulas of Mulinacci (2014),
its Kendall's tau is known to be
\begin{equation}\label{kendall}
\tau_{j,k}=\frac{\theta
-1}\theta+\frac{\tau^{MO}_{j,k}}{\theta}
\end{equation}
where
$$\tau^{MO}_{j,k}=\frac{\alpha_j\alpha_k}{\alpha
_j+\alpha_k-\alpha_j\alpha_k}$$ is the Kendall's tau of the
Marshall-Olkin copula.

Now, the dependence between each default time and the time of a
systemic shock is linear
\begin{equation} \label{Cheru}\tau_{0,j} =\frac{\theta -1}\theta + \frac{\alpha_j}{\theta} \end{equation}
This relationship will be used in our estimation strategy in order to
verify the specification of the model.

\subsection{Estimation strategy}\label{Estima}
In the estimation of the model we assume to observe a panel set of
data $\mu_k(t_i)$, representing marginal default intensities of
$k=1,2,\ldots ,d$ components, for $\{t_1,t_2,\ldots ,t_m\}$ dates.
Our task is to estimate the set of $\alpha_k$ parameters,
representing the sensitivity of each obligor to the systemic
shock, and the parameter $\theta$, that measures the degree of
contagion in the system. We also would like to make a check of the
specification of the model.

Since the main feature of our approach is to identify the weight
of the sensitivity to the systemic shock and of the degree of
contagion in the dependence structure of default times, a natural
estimation strategy would be a moment based approach, which
resembles the calibration procedure proposed by Genest and Rivest
(1993). In particular, our model specification based on linear
distortions and Gumbel dependence makes a procedure based on
Kendall's tau calibration very easy.

Since the model is built to be fully characterized by the
bivariate marginals, the estimation is naturally performed by
calibrating the bivariate Kendall's tau statistics of the system.
For each cluster that we expect to be part of the same
exchangeable system, consisting of $d$ units, we calibrate the set
of $d+1$ parameters of our model.

Formally, we first estimate the Kendall's tau statistics of all the
pairs of the sample, and then estimate the set of parameters
${\bf\Theta}= \{\alpha_1,\alpha_2,\ldots,\alpha_d,\theta \}$ by
solving

$${\bf\hat\Theta} = \underset{\{\alpha_1,\alpha_2,\ldots,\alpha_d,\theta
\}}{\operatorname{argmin}} \sum_{i=1}^{d-1}\sum_{j=i+1}^{d} dist(\hat
\tau_{i,j},\tau_{i,j}( \alpha_i, \alpha_j, \theta)) $$ where
$dist(x,y)$ is a suitable distance measure, $\tau_{i,j}(
\alpha_i,\alpha_j, \theta) $ is the theoretical Kendall's tau based on estimates, and $\hat \tau_{i,j}$ is the corresponding
empirical Kendall's tau statistics. As for the parameters set, $\alpha_i$ represents the sensitivity of component $i$ to a
systemic shock, and $\theta$ represents the contagion parameter,
that is assumed to be the same across all pairs.

The structure of the model also provides an easy procedure to
check whether the specification of the model provides a good fit
to the data. The idea is that if the model is well specified we
could use it to estimate the intensity of the systemic shock from
market data, and use that information to directly verify the
specification of the model. The property of exponentially
distributed marginals of the Gumbel specification is particularly
useful in this case. Given a panel of $m$ observations of $d$
intensities, we can estimate a time series of $m$ intensities of
the systemic shock. Using equation (\ref{Marg_Int}) and the
definition of $\alpha_k$ it is straightforward to compute
\begin{equation}\label{Esti_Lambda_0}\hat \lambda_0(t_i) = \frac{\sum_{k=1}^d \mu_k^{\theta}(t_i)}{\sum_{k=1}^d{\frac{1}{ \alpha_k}}} \end{equation}
where $\hat \lambda_0(t_i)$ denotes the estimate of the systemic
shock intensity at time $t_i$.\\ A straightforward visual check of
the specification of the model would then be to estimate the
Kendall's tau value between the arrival time of a systemic shock
and marginal default times. If the model is well specified, the
Kendall's tau values should be aligned on the straight line
described by equation (\ref{Cheru}).\\ In that case, the procedure
also provides a new series representing the implied intensity of
the systemic shock, that may be usefully applied for further
investigation of the cluster and of the system as a whole. As an
example, one could verify whether other elements of the system,
originally not associated to that cluster, actually have the same
dependence with the systemic shock as the other elements of the
cluster. As a second example, one could use the estimated systemic
shock intensities of different clusters to check the degree of
association across clusters. \\ Our application, that is meant to
illustrate this estimation procedure, will be focussed on a set of
European banks. We will assume that the banks of the same country
constitute a cluster, and we will verify in which case this
assumption is borne out by the data.

\section{An extension to hierarchical Archimedean risk factors}\label{extension}
In this section we will consider a possible extension of the model with exchangeable dependence structure presented above.
Clearly, any $d+1$-dimensional copula can be considered in place of the Archimedean one and the same construction implemented.
Among the possible reasonable choices, vine- Archimedean copulas and hierarchical copulas (HAC) could be considered as natural non-exchangeable
extensions. \\
In this paper we will consider $d+1$-dimensional HAC copulas.
These are obtained through the composition of simple Archimedean
copulas: such composition is recursively applied using different
segmentations of the random variables involved. Starting from the
initial variables $u_1,\ldots ,u_{d+1}$, these are grouped in
$l_1$ copulas $C_{1,1},\ldots ,C_{1,l_1}$. Then, these copulas are
grouped in $l_2$ copulas $C_{2,1},\ldots ,C_{2,l_2}$, and up to
the last level where we have just one copula.  In order to ensure
that the so obtained HAC copula is indeed a copula, the generators
$\psi_{i,j}$ of the copulas involved have to be completely
monotone and the same must hold for their compositions
$\psi_{i+1,j}^{-1}\circ \psi_{i,k}$ whenever $C_{i,k}$ is an
argument of $C_{i+1,j}$. When the generators $\psi_{i,j}$ are in
the same parametrized family, the described procedure yields a
copula if inner copulas have a parameter higher than the outer
ones: in this paper we will consider generators belonging to the
same family (see Savu and Trede 2008 and McNeil 2008 amog the
others as references on this topic).
\bigskip

In the \emph{fully nested} case we have
$$C({\it {\bf u}})=C_d\left (\ldots C_3\left (C_2\left (C_1(u_1,u_2),u_3\right ),u_4\right ),\ldots ,u_{d+1}\right ).$$
If the probability distribution of the systemic shock $X_0$
corresponds to $u_1$, then, the idiosyncratic risks $X_i$, $i\geq
1$, can be decreasingly ordered with respect to the dependence to
$X_0$ being
$$C_{X_0,X_i}(u,v)=C_{i-1}(u,v).$$
If, instead, the probability of $X_0$ corresponds to $u_{d+1}$,
then
$$C_{X_0,X_i}(u,v)=C_{d}(u,v).$$
and the dependence structure between each idiosyncratic risk and
the systemic one is the same for all the idiosyncratic triggers.
\\
In the intermediate case in which the probability of $X_0$
corresponds to $u_j$ for some $j=2,\ldots ,d$, we have that
$$C_{X_i,X_0}(u,v)=C_{j-1}(u,v)$$
for those $X_i$ whose probabilities correspond to those $u_i$ with
$i<j$, and
$$C_{X_0,X_i}(u,v)=C_{i-1}(u,v)$$
for the probabilities of those $X_i$ correspond to $u_i$ with
$i>j$.
\medskip

Of course, under other hierarchical configurations, completely
different relationships among  the systemic and the idiosyncratic
risks can be modelled. For example if
$$C({\it {\bf u}})=C\left (C_{h,1}\left (u_1,\ldots ,u_{j-1}\right ),C_{h,2}\left (u_j,u_{j+1},\ldots ,u_{d+1}\right )\right )$$
where $C_{h,1}$ and $C_{h,2}$ are again HAC copulas, and $X_0$ corresponds to $u_j$, we have that
$$C_{X_i,X_0}(u,v)=C(u,v)$$
for all probabilities $X_i$ that correspond to those $u_i$ with
$i<j$ and
$$C_{X_0,X_i}(u,v)=C_{h,2}(u,v)$$
for $X_i$ probabilities that correspond to $u_i$ with $i>j$.
Hence, in the first case, the dependence structure between $X_i$
and $X_0$ is constant and weaker than that in the second case
where however it varies according to the structure of $C_{h,2}$.
\medskip

Notice that, however, whatever is the case, the dependence
structure between $X_0$ and $X_i$ is always Archimedean, exactly
as in the exchangeable case investigated in Section \ref{model}.
As a consequence, the formulas there presented for the Kendall's
tau between the systemic shock and every default time continue to
hold. In particular, if: i) all the copulas involved in the
hierarchical construction are of Gumbel type and ii) for every
idiosyncratic shock arrival time $X_i$ there exists a function
$K_i$ such that $\bar F_0(t)=\psi_\theta (\lambda _{0,i}K_i(t))$
and $\bar F_i(t)=\psi_\theta (\lambda _{i}K_i(t))$,
then (\ref{Cheru}) applies. Moreover, for all those default times
$\tau_j$ such that the corresponding idiosyncratic shock arrival
time $X_i$ has a dependence relationship with the systemic shock
one $X_0$ expressed by the same Gumbel copula 
, the pairs $(\alpha _j,\tau_{0,j})$ must lie on the same straight
line (\ref{Cheru}).\\

Notice, then, between the fully exchangeable system, and the fully
non-exchangeable one, we can identify an intermediate case in
which the exchangeability concept is only applied to the bivariate
relationships between the systemic shock arrival times and the
idiosyncratic shocks, whatever the dependence among the
idiosyncratic shocks could be.
\bigskip

\subsection{Dependence structure of observed default times}
However, since the statistical procedure presented in Section
\ref{Estima} is based on the estimation of the pairwise dependence
structure of the default times $\tau _j$, we will now compute the
Kendall's function and the Kendall's tau of any pair of default
times.
\bigskip

Clearly, the shocks involved are the systemic one and the two idiosyncratic ones that correspond to the default times we are considering. Formally,
let $X_i,X_j,X_k$ be the three shocks arrival times we are considering.
Whatever the hierarchical structure is, their joint survival distribution is of type
$$\bar F(x_i,x_j,x_k)=C_{\psi_\phi}\left (C_{\psi{_\theta}}\left (\bar F_i(x_i),\bar F_j(x_j)\right ),\bar F_k(x_k)\right )$$
where $C_{\psi_\phi}$ and $C_{\psi_\theta}$ are bivariate Archimedean copula functions with generators $\psi_\phi$ and $\psi_\theta$.
\medskip

Here below we drop the notation according to which the systemic shock arrival time is denoted $X_0$, so that we can move it in different places of the hyrarchical
structure. In particular, it is sufficient to study the two cases: the
systemic shock is represented by $X_i$ and the case in which it is
represented by $X_k$. For the sake of simplicity, we will assume
that all marginal survival distributions are differentiable when
needed.
\medskip

\subsubsection{$X_i$ is the arrival time of the systemic shock}
Assume $X_i$ be the systemic shock's arrival time and
$$\tau_j=\min (X_i,X_j),\,\tau_k=\min (X_i,X_k)$$
be the considered default times. Then
$$\bar F_{\tau_j,\tau _k}(t_j,t_k)=
\psi_\phi\left (\psi_\phi^{-1}\circ\psi_\theta\left (\psi_\theta ^{-1}\circ\bar F_i(\max (t_j,t_k))+\psi_\theta ^{-1}\circ\bar F_j(t_j)\right )+\psi_\phi ^{-1}\circ\bar F_k(t_k)\right ),$$
$$\bar F_{\tau _j}(t)=\psi_\theta\left (\psi_\theta ^{-1}\circ\bar F_i(t)+\psi_\theta^{-1}\circ\bar F_j(t)\right )=\psi_\theta \circ H_{0,j}(t)$$
and
$$\bar F_{\tau _k}(t)=\psi_\phi\left (\psi_\phi ^{-1}\circ\bar F_i(t)+\psi_\phi ^{-1}\circ\bar F_k(t)\right )=\psi_\phi \circ H_{0,k}(t)$$
where $H_{0,j}(t)=\psi_\theta ^{-1}\circ\bar F_i(t)+ \psi_\theta ^{-1}\circ\bar F_j(t)$ and $H_{0,k}(t)=\psi_\phi ^{-1}\circ\bar F_i(t)+\psi_\phi ^{-1}\circ\bar F_k(t)$.
Hence, thanks to Sklar's Theorem, from
$$t_j=H_{0,j}^{-1}\circ \psi_\theta^{-1}(u_j)\text{ and }t_k=H_{0,k}^{-1}\circ\psi_\phi^{-1}(u_k)$$
we get that the associated survival copula is
$$\begin{aligned}&\hat C_{\tau_j,\tau _k}(u_j,u_k)=\\
&=\psi_\phi\left (\psi_\phi^{-1}\circ\psi_\theta\left (\psi_\theta ^{-1}\circ\bar F_i(\max(H_{0,j}^{-1}\circ \psi_\theta^{-1}(u_j),H_{0,k}^{-1}\circ\psi_\phi^{-1}(u_k)))+\right .\right .\\
&\left .\left .+\psi_\theta ^{-1}\circ\bar F_j\circ H_{0,j}^{-1}\circ \psi_\theta^{-1}(u_j)
\right )+\psi_\phi ^{-1}\circ\bar F_k\circ H_{0,k}^{-1}\circ\psi_\phi^{-1}(u_k)
\right ).\end{aligned}$$
Set
$$ D_{ij}=\psi_\theta ^{-1}\circ\bar F_i\circ H_{0,j}^{-1},
D_{ik}=\psi_\theta ^{-1}\circ\bar F_i\circ H_{0,k}^{-1},
D_{ji}=\psi_\theta ^{-1}\circ\bar F_j\circ H_{0,j}^{-1},
D_{ki}= \psi_\phi ^{-1}\circ\bar F_k\circ H_{0,k}^{-1}.$$
Then
\begin{equation}\label{c1}\begin{aligned}&\hat C_{\tau_j,\tau _k}(u_j,u_k)=\\
&=\psi_\phi\left (\psi_\phi^{-1}\circ\psi_\theta\left (\max(D_{ij}\circ \psi_\theta^{-1}(u_j),D_{ik}\circ\psi_\phi^{-1}(u_k))+
D_{ji}\circ \psi_\theta^{-1}(u_j)
\right )
+D_{ki}\circ\psi_\phi^{-1}(u_k)
\right )=\\
&=\left\{\begin{array}{c}
\psi_\phi\left (\psi_\phi^{-1}(u_j)+D_{ki}\circ\psi_\phi^{-1}(u_k)\right ), u_k\geq h(u_j)\\

\psi_\phi\left (\psi_\phi^{-1}\circ\psi_\theta\left (D_{ik}\circ\psi_\phi^{-1}(u_k)+
D_{ji}\circ \psi_\theta^{-1}(u_j)
\right )+D_{ki}\circ\psi_\phi^{-1}(u_k)\right ), u_k< h(u_j)\\
\end{array}\right .
\end{aligned}\end{equation}
where
$$h(x)=\psi_\phi\circ D_{ik}^{-1}\circ D_{ij}\circ\psi_\theta^{-1}(x).$$
\medskip

{\bf Restriction on the distribution of $X_i,X_j,X_k$}
\medskip

\noindent Assume that there exist two functions $K$ and $\hat K$ such that
$$\psi_\theta ^{-1}\circ\bar F_i(t)=\hat \lambda _i\hat K(t), \psi_\theta ^{-1}\circ\bar F_j(t)= \lambda _j\hat K(t)$$
and
$$\psi_\phi ^{-1}\circ\bar F_i(t)=\lambda _i K(t),
\psi_\phi ^{-1}\circ\bar F_k(t)=\lambda _k K(t)$$
which implies that
\begin{equation}\label{consistency}\hat K(t)=\frac 1{\hat \lambda _i}\psi^{-1}_\theta\circ\psi_\phi\left (\lambda _iK(t)\right).\end{equation}
Now, setting $\mu_{ij}=\hat\lambda _i+\lambda _j$ and $\mu_{ik}=\lambda _i+\lambda _k$,
$$H_{0,j}(t)=\hat\mu_{ij}\hat K(t)\text{ and }
H_{0,k}(t)=\mu_{ik}K(t)$$
and
$$D_{ij}(x)=\frac {\hat\lambda _i}{\mu_{ij}}x,
D_{ik}(x)=\frac {\hat\lambda _i}{\mu_{ik}}x,
D_{ji}(x)=\frac {\lambda _j}{\mu_{ij}}x,
D_{ki}(x)=\frac {\lambda _k}{\mu_{ik}}x,$$
from which
$$\bar F_{\tau _j}(t)=\psi_\theta \left (\mu _{ij}\hat K(t)\right )\text{ and }
\bar F_{\tau _k}(t)=\psi_\phi \left (\mu _{ik} K(t)\right )$$
and
\begin{equation}\label{c2}\begin{aligned}&\hat C_{\tau_j,\tau _k}(u_j,u_k)=\\
&=\psi_\phi\left (\psi_\phi^{-1}\circ\psi_\theta\left (\max\left (\frac {\hat\lambda _i}{\mu_{ij}} \psi_\theta^{-1}(u_j),\frac {\hat\lambda _i}{\mu_{ik}}\psi_\phi^{-1}(u_k)\right )+
\frac {\lambda _j}{\mu_{ij}} \psi_\theta^{-1}(u_j)
\right )
+\frac {\lambda _k}{\mu_{ik}}\psi_\phi^{-1}(u_k)
\right ).\end{aligned}\end{equation}\medskip

\begin{remark} \label{gumb1}{\bf The Gumbel case}\\
Assume $\psi_\theta (x)=e^{-x^{\frac 1\theta}}$ and $\psi_\phi (x)=e^{-x^{\frac 1\phi}}$, with $\theta\geq\phi\geq 1$. Then $\psi_\phi^{-1}\circ\psi_\theta(x)=x^{\frac \phi\theta}$
and (\ref{c2}) writes
$$\begin{aligned}&\hat C_{\tau_j,\tau _k}(u_j,u_k)=\\
&=\exp\left\{-\left (\left (\max\left (\frac {\hat\lambda _i}{\mu_{ij}} (-\log (u_j))^\theta,\frac {\hat\lambda _i}{\mu_{ik}}(-\log(u_k))^\phi\right )+
\frac {\lambda _j}{\mu_{ij}} (-\log(u_j))^\theta
\right )^{\frac\phi\theta}+\right .\right .\\
&\left .\left .+\frac {\lambda _k}{\mu_{ik}}(-\log (u_k))^\phi
\right )^{\frac 1\phi}\right\}.
\end{aligned}$$
Necessarily, by (\ref{consistency}), $\hat\lambda _i\hat K=\lambda _i^{\frac\theta\phi}K^{\frac\theta\phi}$
and an admissible choice is
$$\hat\lambda_i=\lambda _i^{\frac\theta\phi}\text{ and }\hat K=K^{\frac\theta\phi}.$$
In particular,
if
$K(t)=t^\phi$ and $\hat K(t)=t^\theta$ we recover exponential marginal distributions,
that is
$$\bar F_{\tau _j}(t)=e^{-\mu_{ij}^{\frac 1\theta}t}\text{ and }\bar F_{\tau _k}(t)=e^{-\mu_{ik}^{\frac 1\phi}t}.$$
\end{remark}
\bigskip

{\bf The Kendall's function and Kendall's tau}
\medskip

\noindent\begin{theorem}\label{Teo1}
If $\rho=\psi_\phi^{-1}\circ\psi_\theta$, let (see (\ref{c1}))
$$\begin{aligned}& C(u,v)=\\
&=\left\{\begin{array}{c}
\psi_\phi\left (\psi_\phi^{-1}(u)+D_{ki}\circ\psi_\phi^{-1}(v)\right ), v\geq h(u)\\
\psi_\phi\left (\rho\left (D_{ik}\circ\psi_\phi^{-1}(v)+
D_{ji}\circ \psi_\theta^{-1}(u)
\right )+D_{ki}\circ\psi_\phi^{-1}(v)\right ), v< h(u)\\
\end{array}\right .
\end{aligned}$$
where
$$h(x)=\psi_\phi\circ D_{ik}^{-1}\circ D_{ij}\circ\psi_\theta^{-1}(x).$$
We have that the corresponding Kendall's function $\mathcal K(t)=\mathbb P(C(u,v)\leq t)$ and Kendall's tau are, respectively,
\begin{equation}\label{kendall1}\mathcal K(t)=t-\psi_\phi^\prime\circ \psi_\phi^{-1}(t)\cdot \left [D_{ki}\circ \psi_\phi^{-1}(t)-\int _{D_{ik}\circ\psi_\phi^{-1}(t)}^{D_{ik}\circ G^{-1}\circ\psi_\phi^{-1}(t)}\rho^\prime\circ\rho^{-1}(\psi_\phi^{-1}(t)-D_{ki}\circ D_{ik}^{-1}(z))dz\right ]\end{equation}
and
$$\begin{aligned}\tau&=1+4\int_0^1\psi_\phi^\prime\circ \psi_\phi^{-1}(t)\cdot D_{ki}\circ \psi_\phi^{-1}(t)dt-\\
&-4\int_0^1\psi_\phi^\prime\circ \psi_\phi^{-1}(t)\int _{ D_{ik}\circ\psi_\phi^{-1}(t)}
^{ D_{ik}\circ G^{-1}\circ \psi_\phi^{-1}(t)}\rho^\prime\circ\rho^{-1}(\psi_\phi^{-1}(t)-D_{ki}\circ D_{ik}^{-1}(z))dzdt
\end{aligned}$$
where
$$G(x)=\psi_\phi^{-1}\circ\psi_\theta\circ  D_{ij}^{-1}\circ  D_{ik}(x)+D_{ki}(x).$$
\end{theorem}
\begin{proof} See Appendix \ref{appendix}.
\end{proof}
\bigskip

\begin{remark} {\bf The Gumbel case}\\
In the setting of Remark \ref{gumb1}, we get
$$\mathcal K(t)=t-\frac t\phi(-\log t)^{1-\phi}\cdot\left [\frac{\lambda _k}{\mu_{ik}}(-\log t)^\phi-\frac \phi\theta\int_{\frac{\hat \lambda _i}{\mu_{ik}}(-\log t)^\phi}
^{\frac{\hat \lambda _i}{\mu_{ik}}G^{-1}\left ((-\log t)^\phi\right )}\left ((-\log t)^\phi-\frac {\lambda _k}{\hat\lambda _i}z\right )^{1-\frac \theta\phi}dz\right ]
$$
and
$$\tau=1+\frac{\lambda _k}{\mu_{ik}}\frac 1\phi-\frac 4\theta\int_0^1t(\log t)^{1-\phi}\left (
\int_{\frac{\hat \lambda _i}{\mu_{ik}}(-\log t)^\phi}
^{\frac{\hat \lambda _i}{\mu_{ik}}G^{-1}\left ((-\log t)^\phi\right )}\left ((-\log t)^\phi-\frac {\lambda _k}{\hat\lambda _i}z\right )^{1-\frac \theta\phi}dz\right )dt$$
where $G(x)=\left (\frac{\mu_{ij}}{\mu_{ik}}\right )^{\frac\phi\theta}x^{\frac \phi\theta}+\frac{\lambda_k}{\mu_{ik}}x$.
\end{remark}
\bigskip

\subsubsection{$X_k$ is the arrival time of the systemic shock}
Here we assume that $X_k$ is the arrival time of the shock and
$$\tau_i=\min(X_i,X_k)\text{ and }\tau_j=\min(X_j,X_k)$$
the considered default times.
Then
$$
\bar F_{\tau_i,\tau_j}(t_i,t_j)=\psi_\phi\left [\psi_\phi^{-1}\circ\psi_\theta\left (\psi_\theta^{-1}(\bar F_i(t_i))+\psi_\theta ^{-1}(\bar F_j(t_j))\right )+\psi_\phi^{-1}(\bar F_k(
\max (t_i,t_j))\right ]
$$
and
$$\bar F_{\tau_i}(t_i)=\psi_\phi\circ H_{0,i}(t_i)\text{ and } \bar F_{\tau_j}(t_j)=\psi_\phi\circ H_{0,j}(t_j).$$
where
$$H_{0,i}=\psi_\phi^{-1}\circ\bar F_i+\psi_\phi^{-1}\circ\bar F_k\text{ and }
H_{0,j}=\psi_\phi^{-1}\circ\bar F_j+\psi_\phi^{-1}\circ\bar F_k.$$
If $\rho=\psi_\phi^{-1}\circ\psi_\theta$,
$$\bar F_{\tau_i,\tau_j}(t_i,t_j)=
\psi_\phi\left [\rho\left (\rho^{-1}\circ \psi_\phi^{-1}\circ\bar F_i(t_i)+
\rho^{-1}\circ \psi_\phi^{-1}\circ\bar F_j(t_j)\right )
+\psi_\phi^{-1}\circ\bar F_k(
\max (t_i,t_j))\right ]$$
and, applying Sklar's Theorem, we recover the associated survival copula is
$$\begin{aligned}
\hat C_{\tau_i,\tau _j}(u_i,u_j)&=\psi_\phi\left [\rho\left (\rho^{-1}\circ \psi_\phi^{-1}\circ\bar F_i
\circ H_{0,i}^{-1}\circ\psi_\phi^{-1}(u_i)+\rho^{-1}\circ \psi_\phi^{-1}\circ\bar F_j
\circ H_{0,j}^{-1}\circ\psi_\phi^{-1}(u_j)\right )
+\right .\\
&\left .+\psi_\phi^{-1}\circ\bar F_k(
\max (H_{0,i}^{-1}\circ\psi_\phi^{-1}(u_i),H_{0,j}^{-1}\circ\psi_\phi^{-1}(u_j)))\right ].
  \end{aligned}
$$
Set
$D_{ik}=\psi_\phi^{-1}\circ\bar F_i\circ H_{0,i}^{-1}$,
$D_{jk}=\psi_\phi^{-1}\circ\bar F_j\circ H_{0,j}^{-1}$,
$ D_{ki}=\psi_\phi^{-1}\circ\bar F_k\circ H_{0,i}^{-1}$
and
$ D_{kj}=\psi_\phi^{-1}\circ\bar F_k\circ H_{0,j}^{-1}$.
It follows
\begin{equation}\label{c3}\begin{aligned}
&\hat C_{\tau_i,\tau _j}(u_i,u_j)=\\
&\psi_\phi\left [\rho\left (\rho^{-1}\circ D_{ik}\circ\psi_\phi^{-1}(u_i)+\rho^{-1}\circ D_{jk}\circ\psi_\phi^{-1}(u_j)\right )
+
\max ( D_{ki}\circ\psi_\phi^{-1}(u_i), D_{kj}\circ\psi_\phi^{-1}(u_j))\right ]=\\
&=\left\{\begin{array}{c}
          \psi_\phi\left [\rho\left (\rho^{-1}\circ D_{ik}\circ\psi_\phi^{-1}(u_i)+\rho^{-1}\circ D_{jk}\circ\psi_\phi^{-1}(u_j)\right )
+ D_{ki}\circ\psi_\phi^{-1}(u_i)\right ],\, u_j\geq h(u_i)\\
\psi_\phi\left [\rho\left (\rho^{-1}\circ D_{ik}\circ\psi_\phi^{-1}(u_i)+\rho^{-1}\circ D_{jk}\circ\psi_\phi^{-1}(u_j)\right )
+ D_{kj}\circ\psi_\phi^{-1}(u_j)\right ],u_j<h(u_i)
         \end{array}\right .
  \end{aligned}
\end{equation}
where
$$h(x)=\psi_\phi\circ  D^{-1}_{kj}\circ D_{ki}\circ\psi_\phi^{-1}(x).$$
\medskip

{\bf Restriction on the distribution of $X_j$, $X_j$, $X_k$}\medskip

\noindent Assume there exists a function $K$ such that  $\psi_\phi^{-1}\circ\bar F_v(x)=\lambda _vK(x)$ for $v=i,j,k$,
and set $\mu_{ik}=\lambda _i+\lambda _k$ and $\mu_{jk}=\lambda _j+\lambda _k$. It follows that
$D_{ik}(x)=\frac{\lambda _i}{\mu_{ik}}x$,
$D_{jk}(x)=\frac{\lambda _j}{\mu_{jk}}x$,
$ D_{ki}(x)=\frac{\lambda _k}{\mu_{ik}}x$
and
$ D_{kj}(x)=\frac{\lambda _k}{\mu_{jk}}x$
and the marginal survival distributions can be written as
$$\bar F_{\tau_s}(t)=\psi_\phi\left (\mu_{sk}K(t)\right ),\, s=i,j$$
while the associated survival copula as
$$\begin{aligned}
&\hat C_{\tau_i,\tau _j}(u_i,u_j)=\\
&\psi_\phi\left [\rho\left (\rho^{-1}\left (\frac{\lambda _i}{\mu_{ik}}
\psi_\phi^{-1}(u_i)\right )+\rho^{-1}\left (\frac{\lambda _j}{\mu_{jk}}\psi_\phi^{-1}(u_j)\right )\right )
+
\max (\frac{\lambda _k}{\mu_{ik}}\psi_\phi^{-1}(u_i),\frac{\lambda _k}{\mu_{jk}}\psi_\phi^{-1}(u_j))\right ].
  \end{aligned}
$$
\begin{remark} {\bf The Gumbel case}\\
Assume $\psi_\theta (x)=e^{-x^{\frac 1\theta}}$ and $\psi_\phi (x)=e^{-x^{\frac 1\phi}}$, with $\theta\geq\phi\geq 1$.
Then $\rho(x)=x^{\frac \phi\theta}$ and
$$\begin{aligned}
&\hat C_{\tau_i,\tau _j}(u_i,u_j)=\\
&\exp\left\{-\left [\left (\left (\frac{\lambda _i}{\mu_{ik}}\right )^{\frac \theta\phi}
(-\log (u_i))^{\theta}+\left (\frac{\lambda _j}{\mu_{jk}}\right)^{\frac\theta\phi}(-\log (u_j))^\theta
\right )^{\frac\phi\theta}
+\right .\right .\\
&\left .\left .+
\max \left(\frac{\lambda _k}{\mu_{ik}}(-\log (u_i))^\phi,\frac{\lambda _k}{\mu_{jk}}(-\log (u_j))^\phi\right )\right ]^{\frac 1\phi}
 \right\} \end{aligned}
$$
while, if $K(t)=t^\phi$, we get exponential marginal distributions
$$\bar F_{\tau_i}(t)=e^{-\mu_{ik}^{\frac 1\phi}t}\text{ and }
\bar F_{\tau_j}(t)=e^{-\mu_{jk}^{\frac 1\phi}t}.$$

\end{remark}
\bigskip

{\bf The Kendall's function and Kendall's tau}\medskip

\noindent\begin{theorem}\label{Teo2}
Let (see \ref{c3})
$$\begin{aligned}&C(u,v)=\\
&=\left\{\begin{array}{c}
          \psi_\phi\left [\rho\left (\rho^{-1}\circ D_{ik}\circ\psi_\phi^{-1}(u)+\rho^{-1}\circ D_{jk}\circ\psi_\phi^{-1}(v)\right )
+D_{ki}\circ\psi_\phi^{-1}(u)\right ],\, v\geq h(u)\\
\psi_\phi\left [\rho\left (\rho^{-1}\circ D_{ik}\circ\psi_\phi^{-1}(u)+\rho^{-1}\circ D_{jk}\circ\psi_\phi^{-1}(v)\right )
+ D_{kj}\circ\psi_\phi^{-1}(v)\right ],\, v<h(u)
         \end{array}\right .
\end{aligned}$$
where
$$h(x)=\psi_\phi\circ D^{-1}_{kj}\circ D_{ki}\circ\psi_\phi^{-1}(x).$$
We have that the Kendall's function $\mathcal K(t)=\mathbb P(C(u,v)\leq t)$ and the Kerndall's tau respectively are
\begin{equation}\label{kendall2}\begin{aligned}\mathcal K(t)&=t+
\psi_\phi^\prime\circ \psi_\phi^{-1}(t)\cdot\\
&\cdot \left [
\int_{\rho^{-1}\circ D_{jk}\circ\psi_\phi^{-1}(t)}^{\rho^{-1}\circ \psi_\phi^{-1}\circ\bar F_j\circ G^{-1}\circ\psi_\phi^{-1}(t)}\rho^\prime\circ\rho^{-1}\left\{\psi_\phi^{-1}(t)-\psi_\phi^{-1}\circ\bar F_k
\circ\bar F_j^{-1}\circ \psi_\phi\circ\rho(z)\right\}dz+\right .\\
&\left .+\int_{\rho^{-1}\circ D_{ik}\circ\psi_\phi^{-1}(t)}^{\rho^{-1}\circ \psi_\phi^{-1}\circ\bar F_i\circ G^{-1}\circ\psi_\phi^{-1}(t)}\rho^\prime\circ\rho^{-1}\left\{\psi_\phi^{-1}(t)-
\psi_\phi^{-1}\circ\bar F_k\circ\bar F_i^{-1}\circ\psi_\phi\circ\rho(z)\right\}dz+\right.\\
&\left .-(D_{kj}+D_{ki})\circ\psi_\phi^{-1}(t)+2\psi_\phi^{-1}\circ\bar F_k\circ G^{-1}\circ\psi_\phi^{-1}(t)
\right ]\end{aligned}\end{equation}
and
$$\begin{aligned}\tau
&=1-4\int_0^1\psi_\phi^\prime\circ \psi_\phi^{-1}(t)\cdot
\left [
\int_{\rho^{-1}\circ D_{jk}\circ\psi_\phi^{-1}(t)}^{\rho^{-1}\circ \psi_\phi^{-1}\circ\bar F_j\circ G^{-1}\circ\psi_\phi^{-1}(t)}\rho^\prime\circ\rho^{-1}\left\{\psi_\phi^{-1}(t)-\psi_\phi^{-1}\circ\bar F_k
\circ\bar  F_j^{-1}\circ\psi_\phi\circ\rho(z)\right\}dz+\right .\\
&\left .+\int_{\rho^{-1}\circ D_{ik}\circ\psi_\phi^{-1}(t)}^{\rho^{-1}\circ \psi_\phi^{-1}\circ\bar F_i\circ G^{-1}\circ\psi_\phi^{-1}(t)}\rho^\prime\circ\rho^{-1}\left\{\psi_\phi^{-1}(t)-\psi_\phi^{-1}\circ\bar F_k
\circ\bar F_i^{-1}\circ\psi_\phi\circ\rho(z)\right\}dz\right ]dt+\\
&-4\int_0^1\psi_\phi^\prime\circ \psi_\phi^{-1}(t)\cdot (2\psi_\phi^{-1}\circ\bar F_k\circ G^{-1}-
(D_{kj}+D_{ki}))\circ \psi_\phi^{-1}(t)dt
\end{aligned}$$
where
\begin{equation}\label{g2}
G(z)=\rho\left\{\rho^{-1}\circ \psi_\phi^{-1}\circ\bar F_i(z)+\rho^{-1}\circ \psi_\phi^{-1}\circ\bar F_j
(z)\right\}+\psi_\phi^{-1}\circ\bar F_k(z).\end{equation}
\end{theorem}
\begin{proof} See Appendix \ref{appendix}.\end{proof}
\bigskip

\begin{remark} {\bf The Gumbel case}\\
In the setting of Remark \ref{gumb1} we get
$$\begin{aligned}\mathcal K(t)&=
t-\frac t\phi (-\log t)^{1-\phi}\cdot\left [\frac \phi\theta\int_{\left (\frac{\lambda_j}{\mu_{jk}}\right )^{\frac \theta\phi}(-\log t)^\theta}^{\lambda_j^{\frac \theta\phi}\left (G^{-1}((-\log t)^\phi)\right )^{\frac\theta\phi}}\left ((-\log t)^\phi-\frac{\lambda _k}{\lambda _j}z^{\frac\phi\theta}\right )^{1-\frac\theta\phi}dz+\right .\\
&\left .+\frac \phi\theta\int_{\left (\frac{\lambda_i}{\mu_{ik}}\right )^{\frac \theta\phi}(-\log t)^\theta}^{\lambda_i^{\frac \theta\phi}\left (G^{-1}((-\log t)^\phi)\right )^{\frac\theta\phi}}\left ((-\log t)^\phi-\frac{\lambda _k}{\lambda _i}z^{\frac\phi\theta}\right )^{1-\frac\theta\phi}dz+\right .\\
&\left .-\left (\frac {\lambda _k}{\mu_{ik}}+\frac{\lambda _k}{\mu_{jk}}\right )(-\log t)^\phi+2\lambda _kG^{-1}((-\log t)^\phi)
\right ]
\end{aligned}$$
and
$$\begin{aligned}\tau&=
1-\frac 4\theta\int_0^1t(-\log t)^{1-\phi}\cdot\left [
\int_{\left (\frac{\lambda_j}{\mu_{jk}}\right )^{\frac \theta\phi}(-\log t)^\theta}^{\lambda_j^{\frac \theta\phi}\left (G^{-1}((-\log t)^\phi)\right )^{\frac\theta\phi}}\left ((-\log t)^\phi-\frac{\lambda _k}{\lambda _j}z^{\frac\phi\theta}\right )^{1-\frac\theta\phi}dz+\right .\\
&\left .+\int_{\left (\frac{\lambda_i}{\mu_{ik}}\right )^{\frac \theta\phi}(-\log t)^\theta}^{\lambda_i^{\frac \theta\phi}\left (G^{-1}((-\log t)^\phi)\right )^{\frac\theta\phi}}\left ((-\log t)^\phi-\frac{\lambda _k}{\lambda _i}z^{\frac\phi\theta}\right )^{1-\frac\theta\phi}dz
\right ]dt+\\
&-\frac 4\phi\int_0^1t(-\log t)^{1-\phi}\left (2\lambda _kG^{-1}((-\log t)^\phi )-\left (\frac {\lambda _k}{\mu_{ik}}+\frac {\lambda _k}{\mu_{jk}}\right )(-\log t)^\phi\right )dt
\end{aligned}$$
where
$$G(z)=\left (\lambda _i^{\frac\theta\phi}+\lambda _j^{\frac\theta\phi}\right )^{\frac \phi\theta}z^\phi+\lambda _k z.$$
\end{remark}

\section{An application to the European banking sector}\label{sector}
In this section we apply the model to the issue of evaluating
systemic risk and contagion in a set of European banking systems.
Until now, the European banking system has been segmented at the
national level, and only after November 4th 2014 it is unified
under a common European regulation and supervision setting (the so
called SSM, Single Supervisory Mechanism, see, for example, Ferran
and Babis, 2013). It is then important to recognize the relevance
of systemic risks and contagion at the national level, and address
the issue whether they co-move at the cross-country level. The
task is to check whether the exchangeable contagion model may
provide a good representation of the data. Of course, here our
interest is mainly in the illustration of the estimation technique
and how it can provide a guide for the specification of the model.

\subsection{Data}

We apply the model described above to a sample of 35 banks
representative of 8 countries of the Euro area. The sample used is
the same as in Baglioni and Cherubini (2013). While we refer the
reader to that paper for an in-depth description of the data set,
here we simply mention that the sample consists of those major
European banks that were subject to the stress test exercise in
2012 and for which a time series of CDS quotes was available on
Datastream. The sample consists of daily data of CDS quotes,
ranging from January 2007 to end of August 2012, with the
exception of Greece for which the sample begins on September 21st
2009, Portugal and Spain, for which the sample starts in January
and February 2008, respectively. The survival probabilities were
extracted from the 5 year CDS quote using what is called the
"simple rule", that is assuming a flat default intensity, which is
consistent with the model described in Section \ref{exponential}.
Moreover, since it is well known that data extracted from market
prices embed a risk premium, that is are computed under the risk
neutral measure, we changed the default probabilities by applying
the Sharpe ratio, according to the technique used by the Moody's
rating agency (see Dwyer et al., 2010).

Future research could investigate further the marginal structure
including more sophisticated technologies to "bootstrap" (in the
financial literature meaning of the term) the term structure of
default intensities (Hull and White, 2000).

\subsection{Estimation and results}
The estimation procedure applied to the data was the same
described in Section \ref{Estima}. Namely, we computed pairwise
Kendall's tau values for the survival probabilities of all the
banks in the same country. Then, for each country we estimated the
$\alpha_k$ parameters and the $\theta$ parameter, minimizing the
distance between theoretical and sample Kandall tau's. In our
specific application, we used the quadratic distance. Due to the
presence of local minima, that arose in preliminary work, the
analysis was finally carried out using a standard global
optimization technique, namely simulated annealing.

Our estimation strategy consisted of three steps.
\begin{itemize}
    \item We first
estimated the model on the whole sample for each country.
    \item Then,
for each country we used the model to estimate the intensity of
the systemic shock and we computed the Kendall's tau between the
survival probability of each bank and the systemic shock. We
verified for which countries the model specification is consistent
with the data.
    \item For the countries where the model specification was
    considered consistent, we provided an analysis of the
    stability of parameters, by repeating the estimate in a
sequence of rolling windows.
\end{itemize}

In Table \ref{tab1} we report the results for the estimates
carried out over the whole sample. For each country,  we report:
i) the $\alpha_k$ parameters for each bank; ii) the average
$\bar{\alpha}$ for the country; iii) the contagion parameter
$\theta$ for the country. It is worth mentioning that the average
parameter $\bar{\alpha}$ is computed as the harmonic mean of the
$\alpha_k$ for each country.
\begin{center}
 \begin{table}
\begin{center}\begin{tabular}{|c|c|}
\hline
{\bf GERMANY}&{\bf SPAIN}\\
 \hline
\begin{tabular}{c|c}
   {\bf Bank}&${\bf \alpha}$\\
  \hline
 HSH  &0.288926\\
   \hline
  WEST LB&0.783299\\
  \hline
  POSTBANK&0.782422\\
  \hline
  DZ BANK&0.637563\\
  \hline
  BAYERN LB&0.885879\\
  \hline
  COMMERZ&0.869757\\
  \hline
  DB&0.752793\\
  \hline
   $\theta$=1&$\bar\alpha$=0.625488
   \end{tabular}
&\begin{tabular}{c|c}
   {\bf Bank}&${\bf \alpha}$\\
  \hline
 PASTOR &0.295905\\
   \hline
   BINTEL&0.397199\\
   \hline
   SABADEL&0.699246\\
   \hline
  POPULA& 0.842172\\
  \hline
  CAJA MADRID&0.059619\\
  \hline
  BBVA&6.99$\cdot 10 ^{-7}$\\
  \hline
  SANTANDER&8.44$\cdot 10^{-7}$\\
  \hline
   $\theta$=5.803539&$\bar\alpha$=0.197725
   \end{tabular}
\\
\hline
{\bf ITALY}&{\bf NETHERLANDS}\\
 \hline
\begin{tabular}{c|c}
   {\bf Bank}&${\bf \alpha}$\\
  \hline
  UBI&0.260332\\
   \hline
   M-PASCHI&0.260332\\
   \hline
   INTESA&1.0\\
   \hline
   UNICREDIT&0.583622\\
   \hline
   $\theta$=5.589623&$\bar\alpha$=0.449591
   \end{tabular}
&\begin{tabular}{c|c}
  {\bf Bank}&${\mathbf \alpha}$\\
  \hline
SNS&0.705556\\
\hline
ABN AMRO&1.279263\\
\hline
RABOBANK&0.824459\\
\hline
ING&0.550825\\
\hline
$\theta$=1.279262&$\bar\alpha $=0.674809
\\
  \end{tabular}\\
\hline
{\bf FRANCE}&{\bf GRECE}\\
 \hline
\begin{tabular}{c|c}
   {\bf Bank}&${\bf \alpha}$\\
  \hline
SOC GEN&0.690164\\
\hline
CA&0.469473\\
\hline
BNP&0.502889\\
  \hline
   $\theta$=6.060185&$\bar\alpha$=0.5388

   \end{tabular}
&\begin{tabular}{c|c}
   {\bf Bank}&${\bf \alpha}$\\
  \hline
ALPHA&0.987782\\
\hline
EFG&0.199412\\
\hline
NBG&0.726229\\
  \hline
$\theta$=3.311887&$\bar\alpha$=0.405181
   \end{tabular}
\\
\hline
{\bf PORTUGAL}&{\bf UK}\\
\hline
\begin{tabular}{c|c}
   {\bf Bank}&${\bf \alpha}$\\
  \hline
ESP. SANTO&0.791219\\
\hline
BCP&0.695637\\
\hline
CAIXA GERAL&0.215709\\
  \hline
$\theta$=5.369255&$\bar\alpha$=0.408871
   \end{tabular}
&\begin{tabular}{c|c}
  Bank&$\alpha$\\
\hline
LLOYDS& 0.930215\\
\hline
BARCLAYS&2.88$\cdot 10^{-07}$\\
\hline
HSBC&4.08$\cdot 10^{-18}$\\
\hline
RBS&0.769870\\
\hline
$\theta$=5.369255&$\bar\alpha $=0.842481
\\
 \end{tabular}\\
\hline
\end{tabular}\end{center}
\caption{\footnotesize Parameters' values for different banks and countries}\label{tab1}
\end{table}
\end{center}

Before discussing the parameters estimated, we use them to provide
a visual check of the specification of the model. For each country
we provided a diagram in which on the horizontal axis we reported
the $\alpha_k$ and on the vertical one the Kendall's tau value. In
the diagram we plot the estimated Kendall's tau statistics between
the survival function of a systemic shock, estimated from the
intensity in equation (\ref{Esti_Lambda_0}), and the survival
probabilities of the banks in the sample. In the plot we also
reported the straight line on which the Kendall's tau values
should lie if the model is well specified, according to equation
(\ref{Cheru}).

Figure \ref{fig1} shows that the model provides a good specification for
Portugal, France, Spain and the Netherlands. In particular, the
specification looks very good for Portugal and France. Spain and
the Netherlands provide an interesting insight on the model. In
both these countries there are banks whose value of $\alpha_k$ is
very close to zero. They are Abn Amro in the Netherlands and
Santander and BBVA in Spain. These three banks are not affected by
the systemic shock, and they are very close to the intercept of
the line. Nevertheless, they are linked by a substantial degree of
dependence to the systemic risk factor, meaning that an increase
of their probability of default may be associated to a higher
probability
of a country-wide shock.
\begin{figure}[htbp]
 \centering
 \includegraphics[width=3cm, height=4cm,angle=-90]{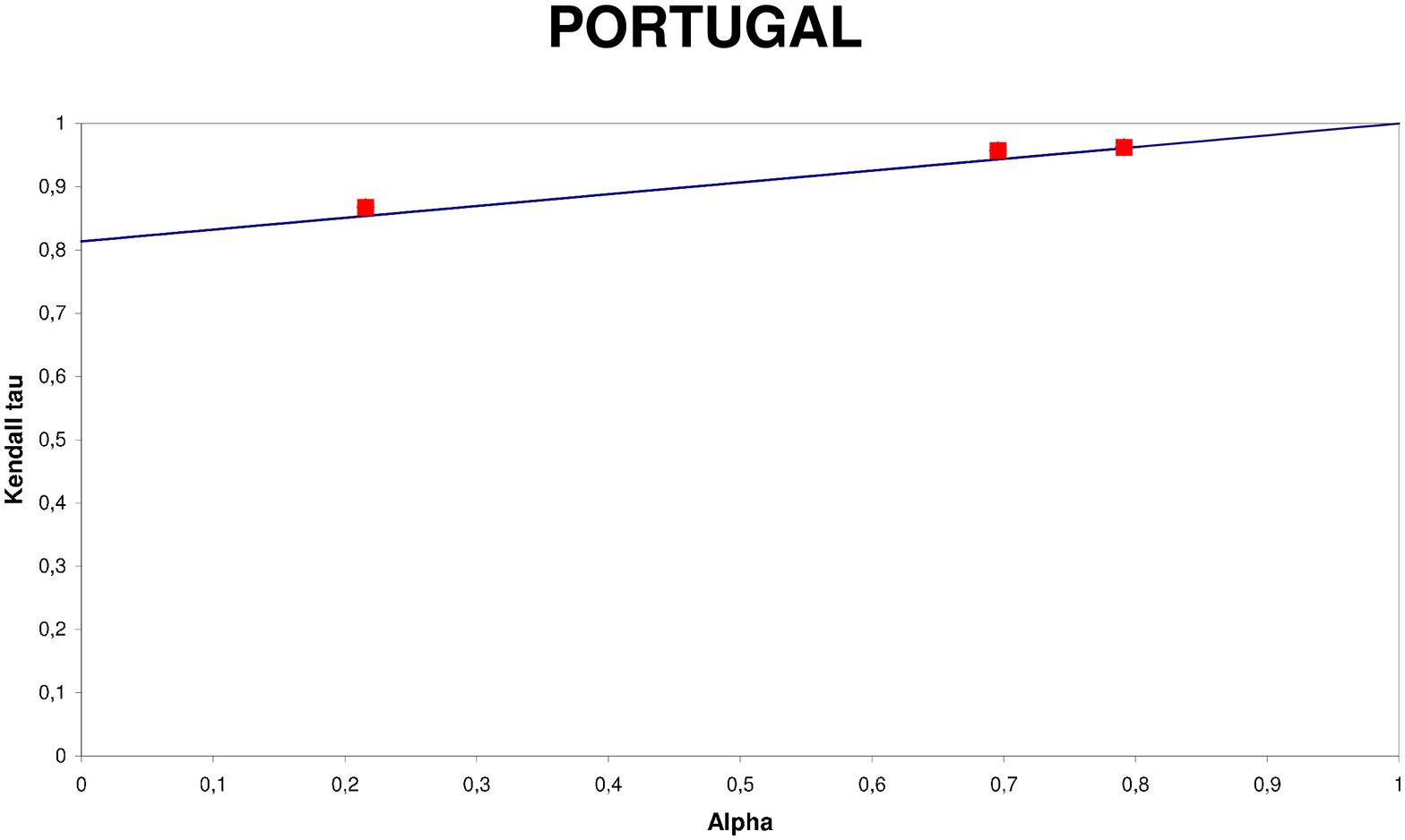}%
\qquad\qquad \includegraphics[width=3cm, height=4cm,angle=-90]{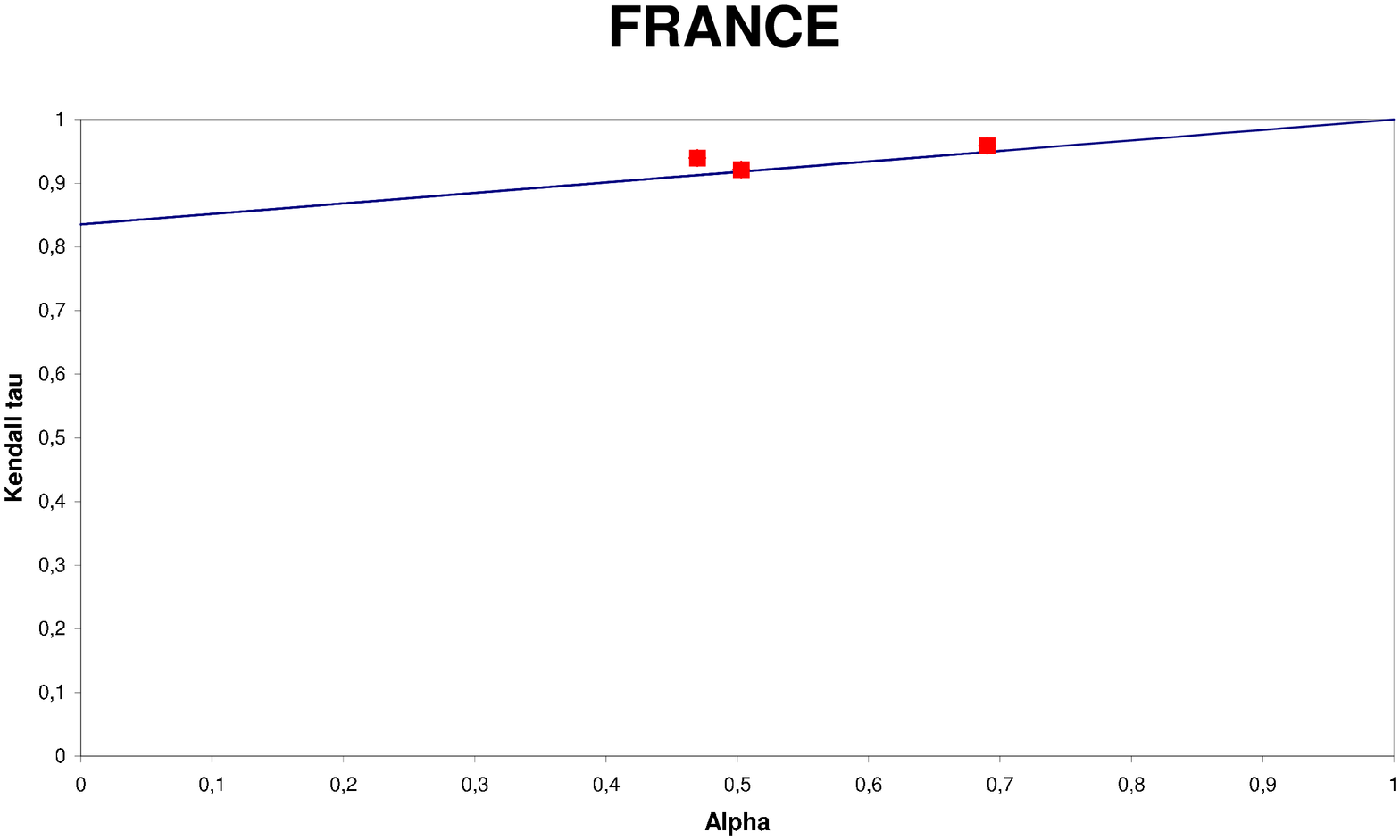}\\
\vskip 0.8cm\includegraphics[width=3cm, height=4cm,angle=-90]{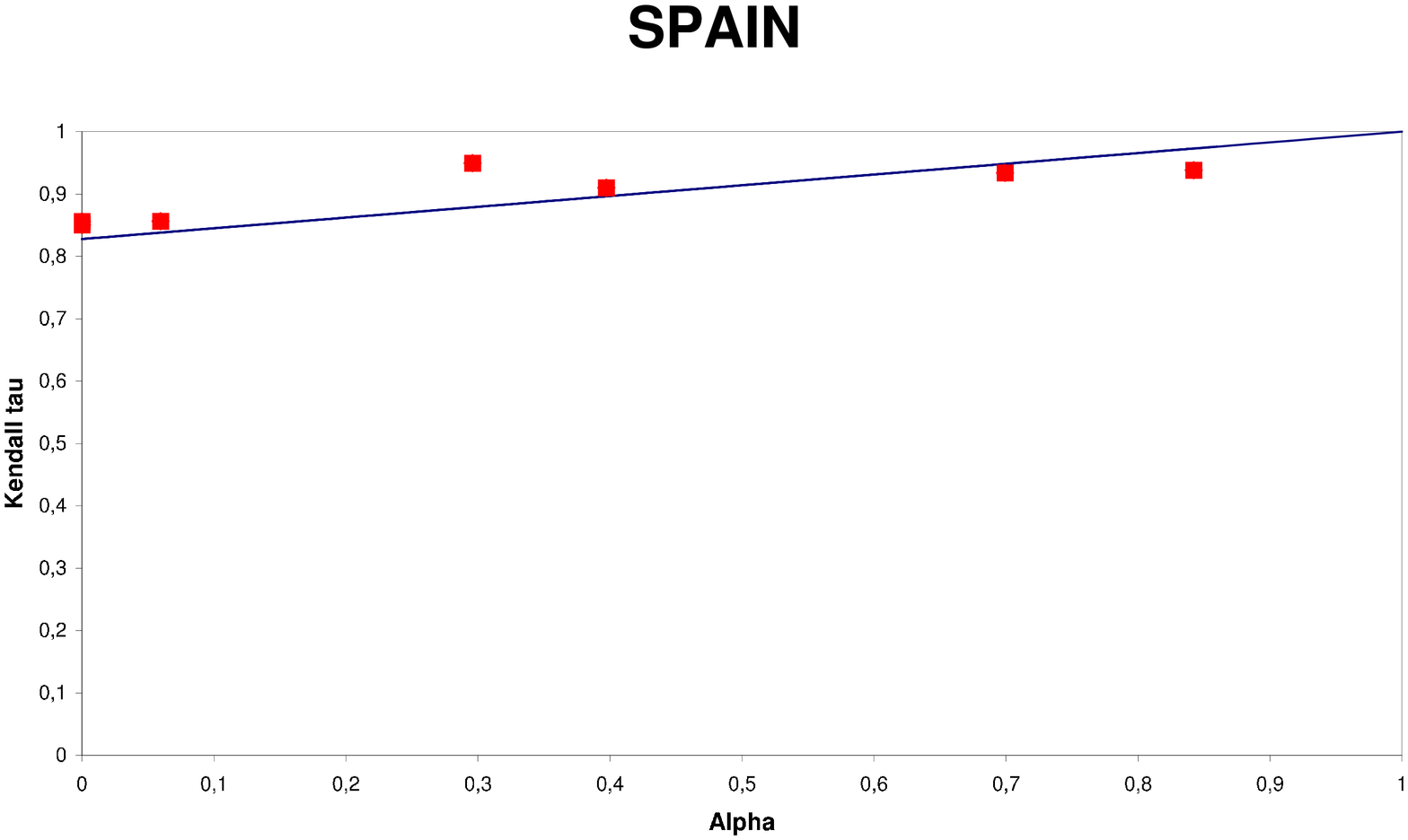}%
\qquad\qquad \includegraphics[width=3cm, height=4cm,angle=-90]{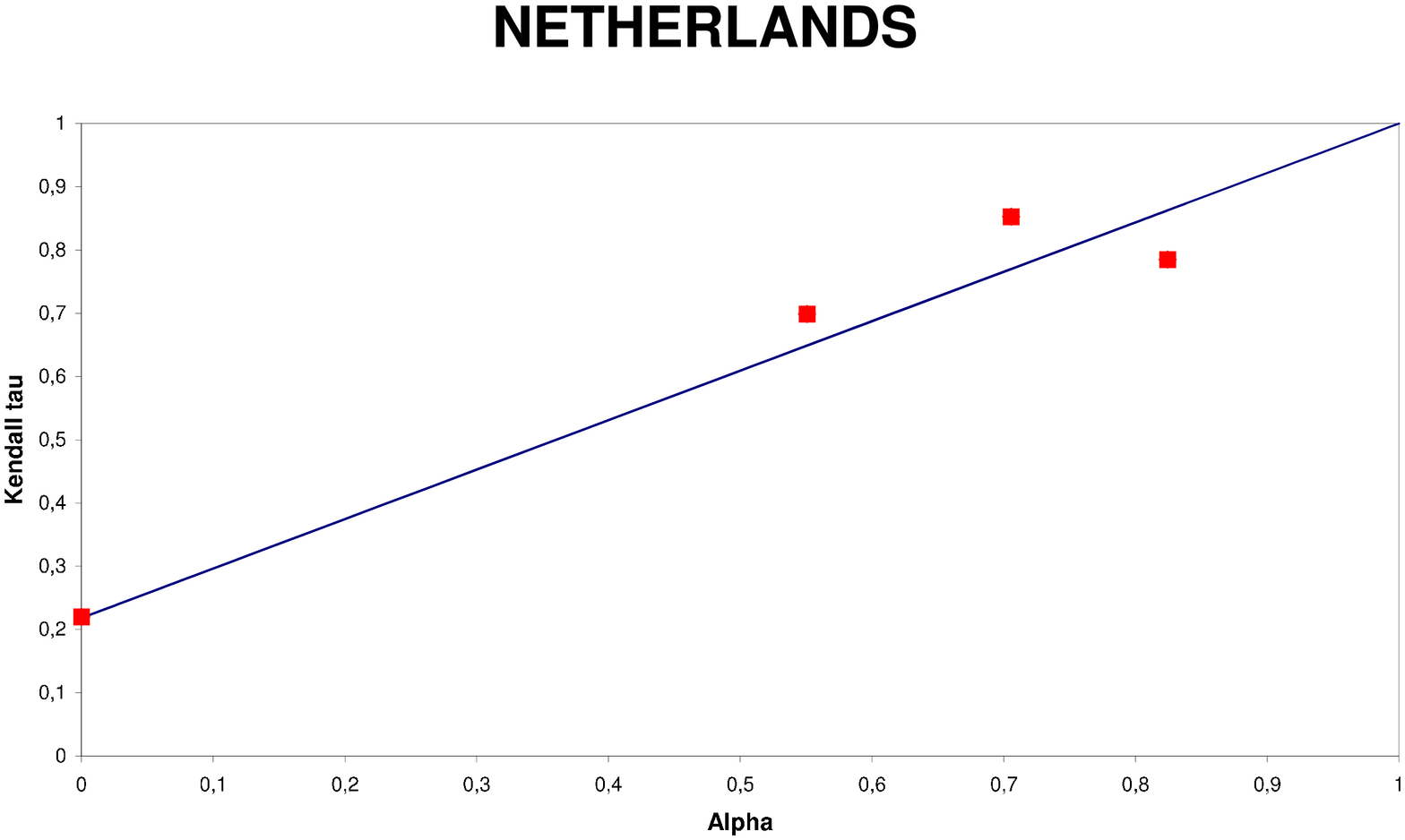}
\vskip 0.5cm
\caption{Kendall's tau between default times and systemic shock}\label{fig1}
\end{figure}
\\
For the other countries (see Figure \ref{fig2}), that is Italy, Greece, UK and Germany,
the model does not seem to fit the data well. In particular,
Italy, Greece and UK could have a chance of a better fit if a
non-exchangeable model were used. For Germany, instead, the model
appears completely wrong, since almost all the Kendall's tau's lie
below the diagonal in a region, and are not even consistent with
the pure systemic risk specification. So, in this case, it seems
that either a model of bivariate relationships without any
systemic shock could be preferable, or that there can be some
other systemic risk factor missing. Actually, the discussion
appeared in newspapers and magazines during the year of the stress
testing analysis would suggest that the German banking system is
exposed to two key risk factors: the first is the exposure to the
so-called "toxic assets", coming from the US subprime crisis, the
second, less known, is exposure to a specific sector of obligors,
namely those linked to the shipping
business.

\begin{figure}[htbp]
 \centering
 \includegraphics[width=3cm, height=4cm,angle=-90]{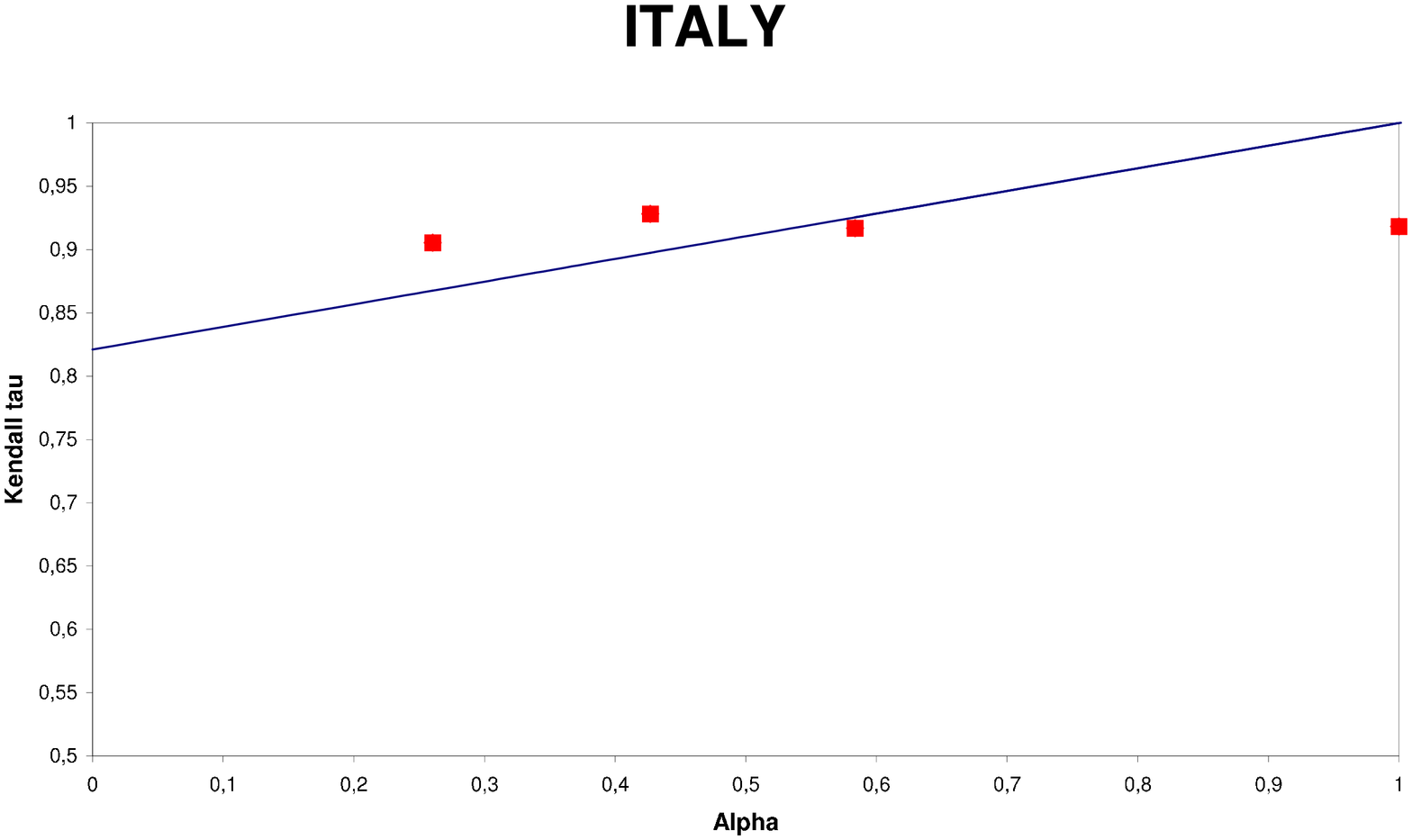}%
\qquad\qquad \includegraphics[width=3cm, height=4cm,angle=-90]{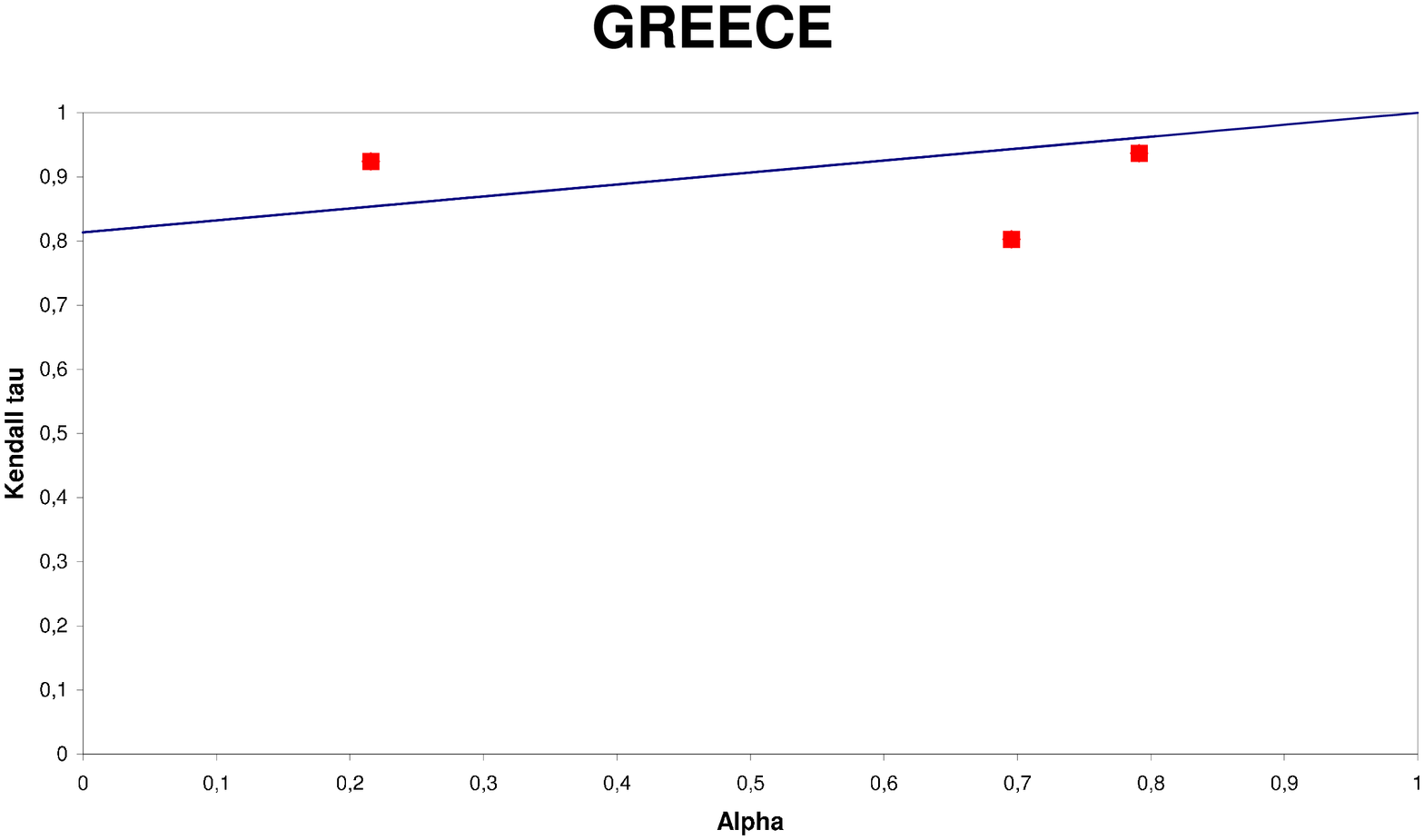}\\
\vskip 0.8cm\includegraphics[width=3cm, height=4cm,angle=-90]{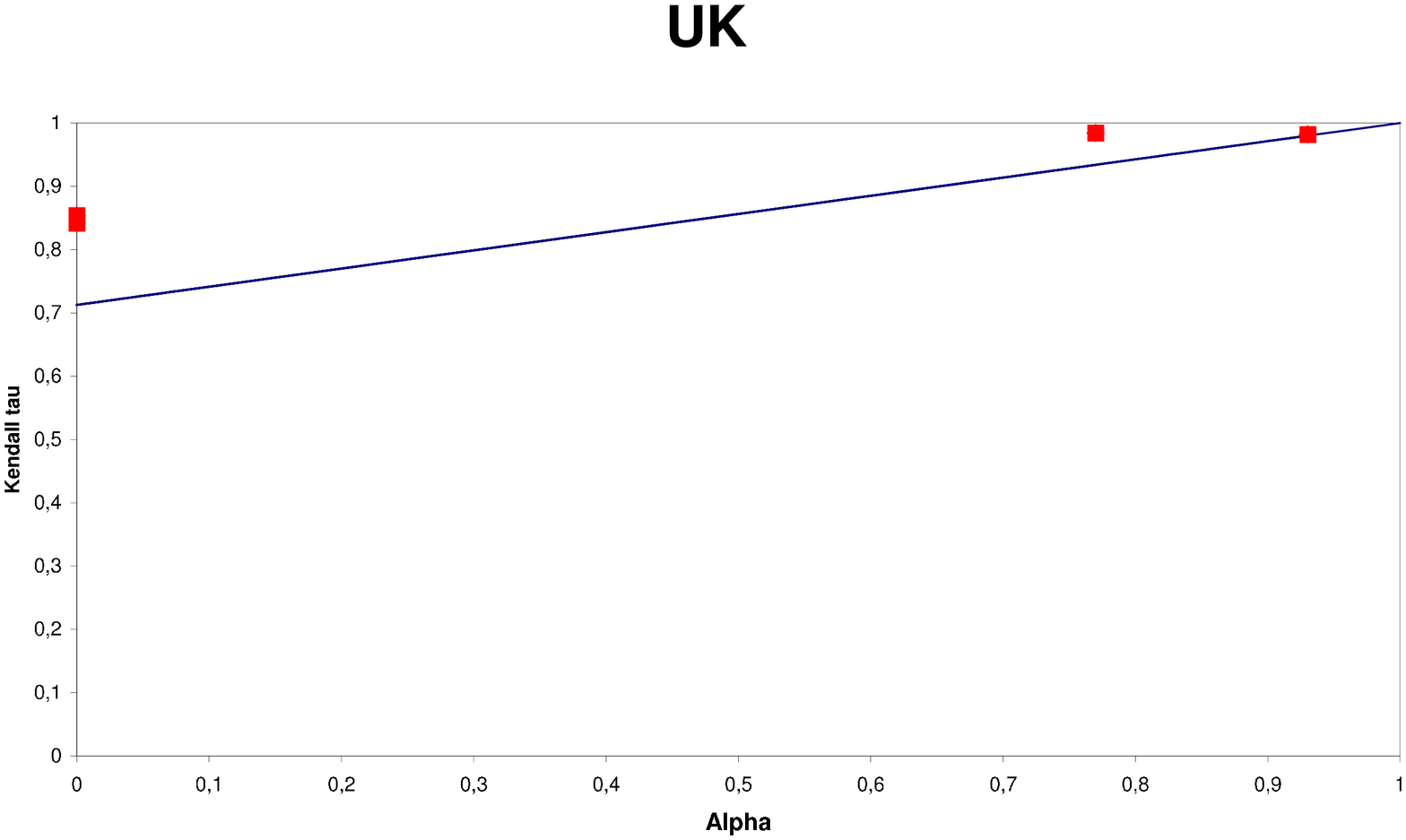}%
\qquad\qquad \includegraphics[width=3cm, height=4cm,angle=-90]{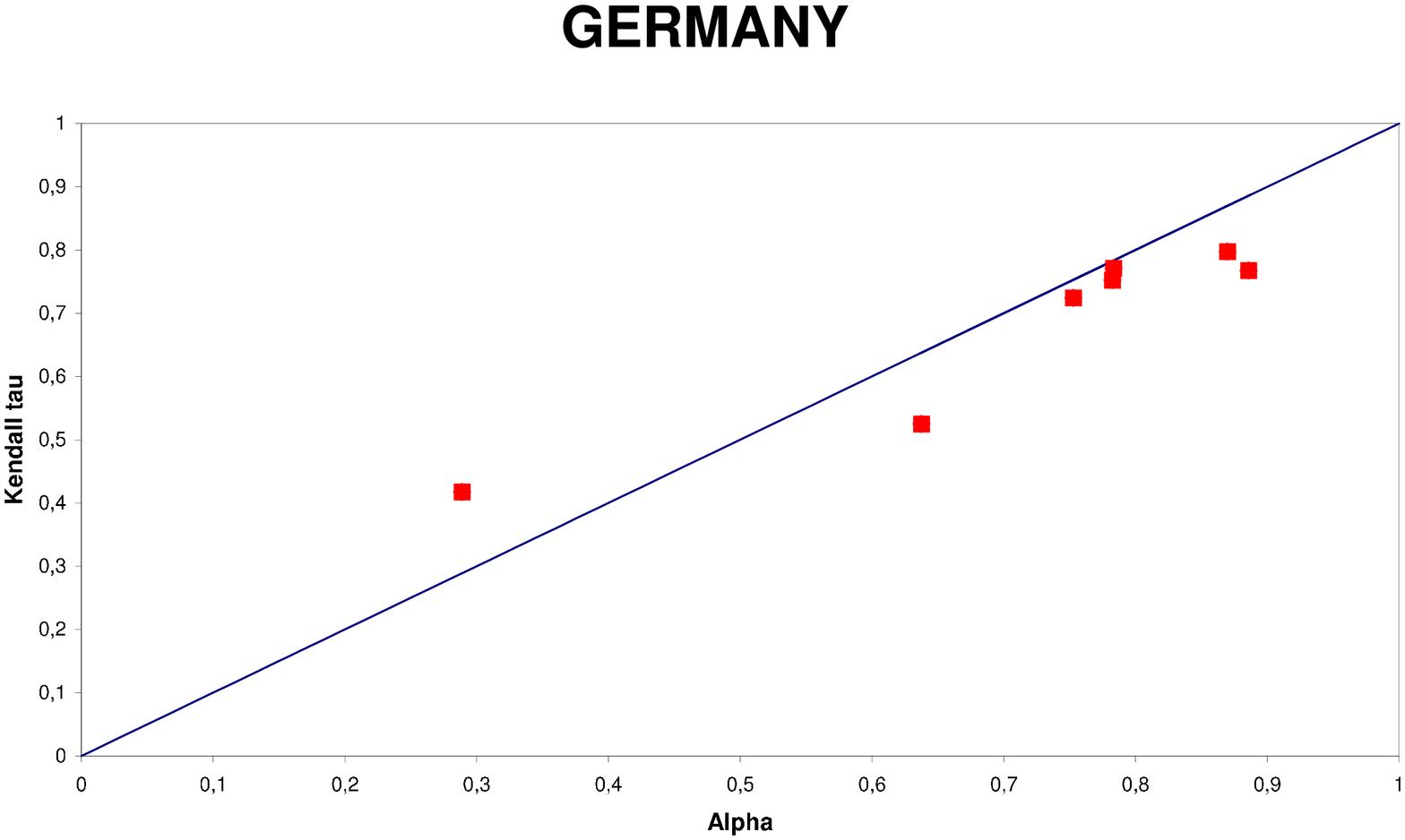}
\vskip 0.5cm
\caption{Kendall's tau between default times and systemic shock}\label{fig2}
\end{figure}

So, our estimation strategy proved able to discriminate cases in
which the model provides a good fit to the data from cases in
which it does not. For the four cases in which the model seems to
work for the entire sample, we now provide an analysis of the
stability of parameters, and in particular of the contagion
parameter $\theta$ across the sample. We replied the estimation
using rolling windows of several lengths, even though here in
order to save space we only report the one based on one year of
daily data. An alternative more sophisticated approach would be to
consider time varying parameters with an estimate performed on
GARCH filtered residuals. This on one side could be more accurate,
while on the other
side it would be inconsistent with the flat intensity assumption, calling for a proper specification of a double stochastic model for each marginal intensity curve. \\
We also performed the analysis first letting all the parameters
change through time, and then assuming the $\alpha_k$ fixed across
the sample, allowing only the contagion parameter $\theta$ to
change. The reason for the latter choice is twofold. First, since
the sensitivity of each bank to shocks mostly depends on its
balance sheet, it is reasonable to assume that the parameters
$\alpha_k$ remain quite stable across the sample. Second, it was
interesting to check whether the estimation of the contagion
parameter, that is the main target of our research, was affected
by changes in the bank specific parameters. We found that the
dynamics of the contagion parameter is almost
indistinguishable in the two cases.\\
In Figure \ref{fig3} we report the results of the analysis for the four
cases in which the model works. The question we have in mind is
whether the contagion parameters increased in the two crucial
periods of the crisis. The first was in the first quarter of 2009,
when the Lehman crisis of September 2008 propagated to Europe. The
second is the sovereign debt crisis triggered by Greece in 2010
and then spread to the other countries of Southern Europe. Figure \ref{fig3}
confirms an interesting co-movement behavior of the contagion
parameters for Spain, Portugal, France and Greece. Intuitively,
when financial crisis spread in the international environment, the
relevance of contagion from the banks within each country is
increasing. Differently from this evidence, however, in the last
part of the sample, characterized by the Italian sovereign crisis,
only contagion within the French banking system seems to markedly
increase, while in the other countries it remains stable or
decrease. This could be consistent with the greater involvement of
the French banking system with the Italian one. In fact, in a
previous version of this work, in which the dynamic analysis had
been carried out for Italy as well, the contagion parameter for
the Italian market was in that period almost indistinguishable
from that of the French one.\bigskip

\begin{figure}[h]
\begin{center}
\includegraphics[width=5cm, height=8cm,angle=-90]{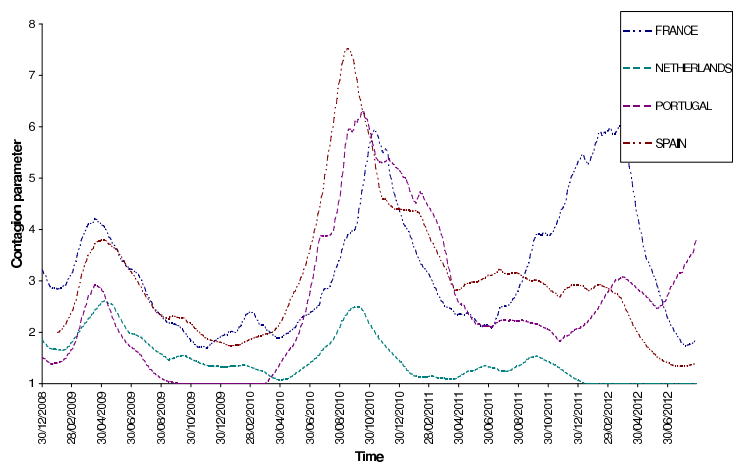}
\vskip 1cm
\caption{The dynamics of the contagion parameter.}\label{fig3}
\end{center}
\end{figure}\bigskip

\bigskip

\bigskip

\section{Conclusions and future extensions}\label{conclusion}
In this paper we presented a model that includes both systemic risk and contagion. Systemic risk is represented by the presence of a shock that brings about the default of all the elements in a cluster. Contagion is represented by the links between the idiosyncratic shocks specific to each component and the systemic shock.\\
On theoretical grounds, the analysis can be carried out assuming
whatever dependence structure among the non observed components
representing the shocks.  On empirical grounds, here we provide an
estimation procedure for a model in which the dependence structure
of the unobserved components is Archimedean and exchangeable.
Moreover, we provide a technique to verify whether the
specification proposed fits the data. We also show that, including
further restrictions may transform the copula model in a new
full-fledged multivariate model with exponential marginal
distributions. \\
Given a panel data of observations of marginal intensities, the
estimation of the model is carried out on the set of bivariate
dependence statistics. Based on estimates, one can extract the
time series of the systemic shock, estimate the Kendall's tau's of
the observed marginals and the systemic shock, and verify whether
they are aligned on a straight line, as predicted by the model. We
apply this technique to a set of European banks of 8 countries,
assuming a systemic shock at the country level, and we found that
our model turns out to be well specified for 4 countries: Spain,
Portugal, France and The Netherlands. For these countries, we also
report an analysis of the dynamics of the contagion parameter,
providing empirical
evidence of co-movement in periods of international crisis.\\
Of course, the next step of this line of research would call for
estimation of the more general, non-exchangeable setting, that has
been also formalized in this paper. More precisely, we see three
main promising fields of development
\begin{itemize}
    \item Estimating the dependence structure of the unobserved
    components directly. Most likely, this would involve the
    application of Simulated Maximum Likelihood (SML) or similar
    techniques, in which one tries to estimate the parameters by
    simulating data as close as possible to the observed ones.
    Doing this can be very easy or very complex, depending on the
    degree of generality that one is willing to accept. As the
    simplest case, assume one could consistently estimate the
    parameters $\alpha_k$ in our model. In this case, it would
    suffice to estimate the systemic and the idiosyncratic
    components from the data and study the dependence analysis on
    those. Exploiting the invariance property of copulas, one
    could directly obtain a consistent estimate of the contagion
    parameters. As the most complex case, assume that the dependence
    structure of the unobserved components must be handled in full
    generality. In that case, the concept itself of the $\alpha_k$
    parameters would be lost, since there is no guarantee that the
    same
    proportionality between the systemic shock intensity and the
    marginal intensity is maintained through the sample.
    \item In a similar line of research, one could also decide
    whether to focus on the full specification of the
    model, or only in the relationship between the systemic shock
    and the marginals. In the latter case, the dependence
    structure among idiosyncratic components would play the role
    of nuisance parameters. For example, our findings of
    exchangeable contagion could be consistent with a dependence
    structure in which some degree of non-exchangeability is
    present, but it is limited to the idiosyncratic shock
    dependence. Within this framework, estimation techniques such
    as those envisaged above could be used to devise formal tests
    of the weaker concept of exchangeability discussed in the extension of our model, in which only the
    pairwise dependence between the idiosyncratic shock and the
    systemic ones are required to have the same copula.
    \item Finally, on a different line of research, one could use
    the estimation procedure applied in this paper as an
    exploratory tool to identify clusters of components that may
    constitute the same "exchangeable systemic contagion cluster".
    This could be done evaluating the dependence between new
    element and the systemic shock representing a cluster. Or it
    can be obtained by measuring the dependence between the
    systemic shocks of different clusters to evaluate if some of
    them can be merged in a single one.
\end{itemize}
As for our specific application to the banking system, of course,
the main challenge would be to extend the analysis to the new
unified European banking system, represented by the 130 banks that
are since now on under the supervision of the European Central
Bank.

\section{Appendix}\label{appendix}
Proof of Proposition \ref{Mulinacci}
\begin{proof}
By (\ref{marginal}), $t_k=H_{0,k}^{-1}\circ \psi^{-1}(u_k)$. Hence
$$C({\emph{\bf u}})=\psi\left (\psi^{-1}(\bar F_0\left (\max_{1\leq k\leq d}\{ H_{0,k}^{-1}\circ \psi^{-1}(u_k)\}\right ))+
\sum_{k=1}^d\psi^{-1}\circ\bar F_k\circ H_{0,k}^{-1}\circ
\psi^{-1}(u_k)\right )$$ Let
$$A_j=\left\{{\emph{\bf u}}\in[0,1]^d:\max_{1\leq i\leq d}\{ H_{0,i}^{-1}\circ \psi^{-1}(u_i)\}=
H_{0,j}^{-1}\circ \psi^{-1}(u_j)\right \}$$ then
$$\begin{aligned}\hat C({\emph{\bf u}}){\bf 1}_{A_j}({\emph{\bf u}})&=\psi\left (\psi^{-1}(\bar F_0\left ( H_{0,j}^{-1}\circ \psi^{-1}(u_j)\right ))+
\sum_{k=1}^d\psi^{-1}\circ\bar F_k\circ H_{0,k}^{-1}\circ \psi^{-1}(u_k)\right )=\\
&=\psi\left (\psi^{-1}(u_j)+ \sum_{k=1,k\neq j}^d\psi^{-1}\circ
\bar F_k\circ H_{0,k}^{-1}\circ \psi^{-1}(u_k)\right )
\end{aligned}$$
\end{proof}
\bigskip

Proof of Theorem \ref{Teo1}
\begin{proof}
In the sequel we set $\partial _1C(u,v)=\frac{\partial}{\partial u}C(u,v)$ and
$\partial _2C(u,v)=\frac{\partial}{\partial v}C(u,v)$.
\medskip

We want to compute the $C$-measure of the set
$$S_t=\{(u,v)\in [0,1]^2:C(u,v)\leq t\}.$$
Notice that the level curve $C(u,v)=t$ intersects the graph of the function $v=h(u)$ in a unique point that we denote with $(u_t,v_t)$. Hence $S_t$ can be decomposed as $S_t=R_t+ R_{1,t}+R_{2,t}$
where $R_t=[0,u_t]\times [0,v_t]$, $R_{1,t}=\{(u,v): v\in (v_t,1], C(u,v)\leq t\}$ and $R_{2,t}=\{(u,v):u\in (u_t,1], C(u,v)\leq t\}$.\\
Clearly, the $C$-measure of $R_t$ is $t$. In order to compute the $C$-measure of $R_{1,t}$ and $R_{2,t}$, we compute $u_t$ and $v_t$.
Since $(u_t,v_t)$ satisfies $\psi_\phi\left (\psi_\phi^{-1}(u_t)+D_{ki}\circ\psi_\phi^{-1}(v_t)\right )=t$ and $v_t=h(u_t)$, we get
$$\psi_\phi^{-1}\circ\psi_\theta\circ D_{ij}^{-1}\circ  D_{ik}\circ \psi_\phi^{-1}(v_t)+D_{ki}\circ \psi_\phi^{-1}(v_t)=\psi_\phi^{-1}(t)$$
from which
$$v_t=\psi_\phi\circ G^{-1}\circ\psi_\phi^{-1}(t)$$
and
$$u_t=\psi_\theta\circ D_{ij}^{-1}\circ D_{ik}\circ G^{-1}\circ\psi_\phi^{-1}(t).$$
Let us start with $R_{1,t}$. Notice that here, $C(u,v)\leq t$ is equivalent to $u\leq F_1(t,v)$ where $F_1(t,v)=\psi_\phi\left (\psi_\phi^{-1}(t)-D_{ki}\circ\psi_\phi^{-1}(v)\right )$. Hence
$$\begin{aligned}
\mathbb P(R_{1,t})&=\int_{v_t}^1\mathbb P(U\leq F_1(t,v)\vert V=v)dv=\\
&=\int_{v_t}^1\partial _2C(F_1(t,v),v)dv=\\
&=\int_{v_t}^1\psi_\phi^\prime\circ\psi_\phi^{-1}(t)\cdot\frac{d}{dv}D_{ki}\circ\psi_\phi^{-1}(v)dv=\\
&=-\psi_\phi^\prime\circ\psi_\phi^{-1}(t)\cdot D_{ki}\circ\psi_\phi^{-1}(v_t).
\end{aligned}$$
Let us now consider $R_{2,t}$.
Notice that here, the inequality $C(u,v)\leq t$, is equivalent to $u\leq F_2(t,v)$ where $$F_2(t,v)=\psi_\theta\circ D_{ji}^{-1}
\left (\rho^{-1}\left (\psi_\phi^{-1}(t)-D_{ki}\circ\psi_\phi^{-1}(v)\right )- D_{ik}\circ \psi_\phi^{-1}(v)\right ).$$
But $$R_{2,t}=\{(u,v):u_t<u\leq 1, t<v, C(u,v)\leq t\}\cup \{(u,v):u_t<u\leq 1,v\leq t\}$$
and
$$\begin{aligned}&\mathbb P\left ( u_t<U\leq 1,V\leq t\right )=t-C(u_t,t)=\\
&=\mathbb P\left (U\leq u_t,t<V\leq v_t\right ).\end{aligned}$$
Hence
$$\begin{aligned}
&\mathbb P(R_{2,t})=\int_t^{v_t}\mathbb P(U\leq F_2(t,v)\vert V=v)dv=\\
&=\int_t^{v_t}\partial _2C(F_2(t,v),v)dv=\\
&=\int_{t}^{v_t}
\psi_\phi^\prime\circ\psi_\phi^{-1}(t)\{\rho ^\prime\circ\rho ^{-1}\left (
\psi_\phi^{-1}(t)-D_{ki}\circ\psi_\phi^{-1}(v)\right )
\frac{d}{dv} D_{ik}\circ\psi_\phi^{-1}(v)+\frac{d}{dv}D_{ki}\circ\psi_\phi^{-1}(v)\}dv=\\
&=\psi_\phi^\prime\circ\psi_\phi^{-1}(t)\left\{
\int _{ D_{ik}\circ\psi_\phi^{-1}(t)}^{D_{ik}\circ\psi_\phi^{-1}(v_t)}\rho^\prime\circ\rho^{-1}(\psi_\phi^{-1}(t)-D_{ki}\circ D_{ik}^{-1}(z))dz+\right .\\
&\left .+D_{ki}\circ\psi_\phi^{-1}(v_t)-D_{ki}\circ\psi_\phi^{-1}(t)\right\}.
\end{aligned}$$
From $\mathbb P(S_t)=t+\mathbb P(R_{1,t})+\mathbb P(R_{2,t})$ we get (\ref{kendall1}).
\medskip

As a consequence, the Kendall's tau is
$$\begin{aligned}\tau&=3-4\int_0^1\mathcal K(t)dt=\\
&=3-4\int_0^1\left\{t-\psi_\phi^\prime\circ \psi_\phi^{-1}(t)\cdot \right .\\
&\left .\cdot \left [D_{ki}\circ \psi_\phi^{-1}(t)-\int _{ D_{ik}\circ\psi_\phi^{-1}(t)}^{ D_{ik}\circ G^{-1}\circ\psi_\phi^{-1}(t)}\rho^\prime\circ\rho^{-1}(\psi_\phi^{-1}(t)-D_{ki}\circ D_{ik}^{-1}(z))dz\right ]\right\}dt=\\
&=1+4\int_0^1\psi_\phi^\prime\circ \psi_\phi^{-1}(t)\cdot D_{ki}\circ \psi_\phi^{-1}(t)dt-\\
&-4\int_0^1\psi_\phi^\prime\circ \psi_\phi^{-1}(t)\int _{ D_{ik}\circ\psi_\phi^{-1}(t)}
^{ D_{ik}\circ G^{-1}\circ \psi_\phi^{-1}(t)}\rho^\prime\circ\rho^{-1}(\psi_\phi^{-1}(t)-D_{ki}\circ D_{ik}^{-1}(z))dzdt.
\end{aligned}$$
\end{proof}
\bigskip

Proof of Theorem \ref{Teo2}
\begin{proof} In the sequel we set $\partial _1C(u,v)=\frac{\partial}{\partial u}C(u,v)$ and
$\partial _2C(u,v)=\frac{\partial}{\partial v}C(u,v)$.
\medskip

The proof is similar to the one of Theorem \ref{Teo1}.\medskip

Again we decompose the set
$S_t=\{(u,v)\in [0,1]^2:C(u,v)\leq t\}$
as $S_t=R_t+ R_{1,t}+R_{2,t}$
where, if $(u_t,v_t)$ is the intersection point of the curves $C(u,v)=t$ and $v=h(u)$, $R_t=[0,u_t]\times [0,v_t]$, $R_{1,t}=\{(u,v):u\in (u_t,1], C(u,v)\leq t\}$ and $R_{2,t}=\{(u,v): v\in (v_t,1], C(u,v)\leq t\}$.\\
Clearly, the $C$-measure of $R_t$ is $t$.
In order to compute the $C$-measure of $R_{1,t}$ and $R_{2,t}$, we compute $u_t$ and $v_t$.
Since
$$\psi_\phi\left [\rho\left (\rho^{-1}\circ D_{ik}\circ\psi_\phi^{-1}(u_t)+\rho^{-1}\circ D_{jk}\circ\psi_\phi^{-1}(v_t)\right )+ D_{kj}\circ\psi_\phi^{-1}(v_t)\right ]=t$$ and $v_t=h(u_t)$, we get
$$v_t=\psi_\phi\circ H_{0,j}\circ G^{-1}\circ\psi_\phi^{-1}(t)$$
and
$$u_t=\psi_\phi\circ H_{0,i}\circ G^{-1}\circ\psi_\phi^{-1}(t)$$
where $G$ is given by (\ref{g2}).\\
Let us start with $R_{1,t}$.
Notice that here, $C(u,v)\leq t$ is equivalent to $u\leq F_1(t,v)$ where $$F_1(t,v)=\psi_\phi\circ D_{ik}^{-1}\circ\rho \left (\rho^{-1}\left (\psi_\phi^{-1}(t)- D_{kj}\circ\psi_\phi^{-1}(v)\right )-\rho^{-1}\circ D_{jk}\circ \psi_\phi^{-1}(v)\right ).$$
By similar arguments as those used in the proof of Theorem \ref{Teo1}, we have
$$\begin{aligned}
&\mathbb P(R_{1,t})=\int_t^{v_t}\mathbb P(U\leq F_1(t,v)\vert V=v)dv=\\
&=\int_t^{v_t}\partial _2C(F_1(t,v),v)dv=\\
&=\int_t^{v_t}
\psi_\phi^\prime\circ\psi_\phi^{-1}(t)\left\{\rho ^\prime\circ\rho ^{-1}\left (
\psi_\phi^{-1}(t)- D_{kj}\circ\psi_\phi^{-1}(v)\right )
\frac{d}{dv}\rho^{-1}\circ D_{jk}\circ\psi_\phi^{-1}(v)+\right .\\
&\left .+\frac{d}{dv} D_{kj}\circ\psi_\phi^{-1}(v)\right\}dv=\\
&=\psi_\phi^\prime\circ\psi_\phi^{-1}(t)\left\{
\int _{\rho^{-1}\circ D_{jk}\circ\psi_\phi^{-1}(t)}^{\rho^{-1}\circ  D_{jk}\circ\psi_\phi^{-1}(v_t)}
\rho^\prime\circ\rho^{-1}(\psi_\phi^{-1}(t)-\psi_\phi^{-1}\circ\bar F_k\circ\bar F_j^{-1}\circ\psi_\phi\circ\rho(z))dz+\right .\\
&\left .+ D_{kj}\circ\psi_\phi^{-1}(v_t)- D_{kj}\circ\psi_\phi^{-1}(t)\right\}.
\end{aligned}$$
Substituting $v_t$ we get
$$\begin{aligned}
&\mathbb P(R_{1,t})=\\
&=\psi_\phi^\prime\circ\psi_\phi^{-1}(t)\left\{
\int _{\rho^{-1}\circ D_{jk}\circ\psi_\phi^{-1}(t)}^{\rho^{-1}\circ \psi_\phi^{-1}\circ\bar F_j
\circ G^{-1}\psi_\phi^{-1}(t)}\rho^\prime\circ\rho^{-1}(\psi_\phi^{-1}(t)-\psi_\phi^{-1}\circ\bar F_k
\circ\bar F_j^{-1}\circ\psi_\phi\circ\rho(z))dz+\right .\\
&\left .+\psi_\phi^{-1}\circ\bar F_k\circ G^{-1}\circ\psi_\phi^{-1}(t)- D_{kj}\circ\psi_\phi^{-1}(t)\right\}.
\end{aligned}$$
With similar computations we get
$$\begin{aligned}
&\mathbb P(R_{2,t})=\\
&=\psi_\phi^\prime\circ\psi_\phi^{-1}(t)\left\{
\int _{\rho^{-1}\circ D_{ik}\circ\psi_\phi^{-1}(t)}^{\rho^{-1}\circ  \psi_\phi^{-1}\circ\bar F_i
\circ G^{-1}\psi_\phi^{-1}(t)}\rho^\prime\circ\rho^{-1}(\psi_\phi^{-1}(t)-\psi_\phi^{-1}\circ\bar F_k
\circ\bar F_i^{-1}\circ\psi_\phi\circ\rho(z))dz+\right .\\
&\left .+\psi_\phi^{-1}\circ\bar F_k\circ G^{-1}\circ\psi_\phi^{-1}(t)- D_{ki}\circ\psi_\phi^{-1}(t)\right\}.
\end{aligned}$$
From $\mathbb P(S_t)=t+\mathbb P(R_{1,t})+\mathbb P(R_{2,t})$ we get (\ref{kendall2}).
\medskip

As a consequence, the Kendall's tau is
$$\begin{aligned}\tau&=3-4\int_0^1\mathcal K(t)dt=\\
&=1-4\int_0^1\psi_\phi^\prime\circ \psi_\phi^{-1}(t)\cdot \\
&\cdot\left [
\int_{\rho^{-1}\circ D_{jk}\circ\psi_\phi^{-1}(t)}^{\rho^{-1}\circ \psi_\phi^{-1}\circ\bar F_j\circ G^{-1}\circ\psi_\phi^{-1}(t)}\rho^\prime\circ\rho^{-1}\left\{\psi_\phi^{-1}(t)-\psi_\phi^{-1}\circ\bar F_k
\circ\bar  F_j^{-1}\circ\psi_\phi\circ\rho(z)\right\}dz+\right .\\
&\left .+\int_{\rho^{-1}\circ D_{ik}\circ\psi_\phi^{-1}(t)}^{\rho^{-1}\circ \psi_\phi^{-1}\circ\bar F_i\circ G^{-1}\circ\psi_\phi^{-1}(t)}\rho^\prime\circ\rho^{-1}\left\{\psi_\phi^{-1}(t)-\psi_\phi^{-1}\circ\bar F_k
\circ\bar F_i^{-1}\circ\psi_\phi\circ\rho(z)\right\}dz\right ]dt+\\
&-4\int_0^1\psi_\phi^\prime\circ \psi_\phi^{-1}(t)\cdot (2\psi_\phi^{-1}\circ\bar F_k\circ G^{-1}-( D_{kj}+ D_{ki}))\circ \psi_\phi^{-1}(t)dt.
\end{aligned}$$
\end{proof}

\end{document}